\documentclass[english,pra,amsfonts,amssymb,amsmath,twocolumn,groupedaddress]{revtex4-1}
\usepackage[utf8]{inputenc}
\pdfpageattr{/Group <</S /Transparency /I true /CS /DeviceRGB>>}
\usepackage{amsthm}
\usepackage{amsmath}
\usepackage{amssymb}
\usepackage{subcaption}

\PassOptionsToPackage{normalem}{ulem}
\usepackage{ulem}

\usepackage{comment}
\usepackage{caption}
\usepackage{graphicx}

\usepackage{algpseudocode}

\algnewcommand{\A}{\textbf{and}\space}
\algnewcommand{\Or}{\textbf{or}\space}
\algnewcommand{\Xor}{\textbf{xor}\space}
\makeatletter
\let\OldStatex\Statex
\renewcommand{\Statex}[1][3]{%
  \setlength\@tempdima{\algorithmicindent}%
  \OldStatex\hskip\dimexpr#1\@tempdima\relax}
\makeatother

\usepackage[unicode=true,pdfusetitle,
 bookmarks=true,bookmarksnumbered=false,bookmarksopen=false,
 breaklinks=false,pdfborder={0 0 1},backref=false,colorlinks=false]
 {hyperref}
\usepackage[usenames,dvipsnames]{xcolor}
\hypersetup{colorlinks=true, linkcolor=Maroon, citecolor=OliveGreen, filecolor=magenta, urlcolor=Blue}

\makeatletter

\theoremstyle{plain}
\newtheorem{thm}{\protect\theoremname}
\theoremstyle{plain}

\newtheorem{lem}[thm]{Lemma}
\newtheorem{defin}[thm]{\protect\definname}
\newtheorem{observ}[thm]{\protect\observname}
\newtheorem{corol}[thm]{\protect\corolname}

\newtheorem{example}[thm]{\protect\examplename}

\makeatother

\usepackage{babel}

\providecommand{\theoremname}{Theorem}
\providecommand{\definname}{Definition}
\providecommand{\observname}{Observation}
\providecommand{\corolname}{Corollary}
\providecommand{\algorithmname}{Algorithm}
\providecommand{\examplename}{Example}
\providecommand{\problemname}{Problem}

\usepackage{pgfplots}
\usetikzlibrary{calc}
\usepackage{pgfplotstable}
\usepackage{colortbl}
\usepackage{booktabs}
\usepackage{pgfplotstable}

\def\C{{\mathbb{C}}}
\def\Z{{\mathbb{Z}}}
\def\R{{\mathbb{R}}}
\def\Q{{\mathbb{Q}}}
\def\N{{\mathbb{N}}}

\newcommand{\ket}[1]{|#1\rangle}

\newcommand{\onemat}[0]{{\mathbf 1}}
\newcommand{\zeromat}[0]{{\mathbf 0}}
\newcommand{\nix}[1]{{}}

\begin{document}

\title{Efficient Synthesis of Universal Repeat-Until-Success Circuits}

\author{Alex Bocharov$^*$}
\author{Martin Roetteler$^*$}
\author{Krysta M.~Svore$^*$}

\affiliation{
$^*$Quantum Architectures and Computation Group, Microsoft Research, Redmond, WA (USA)}

\begin{abstract}

Recently, it was shown that Repeat-Until-Success (RUS) circuits can achieve a $2.5$ times reduction in expected depth over ancilla-free techniques for single-qubit unitary decomposition.
However, the previously best-known algorithm to synthesize RUS circuits requires exponential classical runtime.
In this work we present an algorithm to synthesize an RUS circuit to approximate any given single-qubit unitary within precision $\varepsilon$ in probabilistically polynomial classical runtime.
Our synthesis approach uses the Clifford+$T$ basis, plus one ancilla qubit and measurement.
We provide numerical evidence that our RUS circuits have an expected $T$-count on average $2.5$ times lower than the theoretical lower bound of $3 \log_2 (1/\varepsilon)$ for ancilla-free single-qubit circuit decomposition.

\end{abstract}

\maketitle

\emph{Introduction.}
With rapid maturation of quantum devices, efficient compilation of high-level quantum algorithms into lower-level fault-tolerant circuits is critical.
A popular fault-tolerant universal quantum basis is the \emph{Clifford+T} basis, consisting of the two-qubit controlled-NOT gate ($\mbox{CNOT}$),
and the single-qubit Hadamard ($H$) and $T$ gates, given by
$ H = \frac{1}{\sqrt{2}} \, \left[\begin{smallmatrix}1&1\\1&\textrm{-}1\end{smallmatrix}\right]$,
and $T = \left[\begin{smallmatrix}1&0\\0&e^{i \pi/4}\end{smallmatrix}\right]$.

Efficient algorithms for approximating a single-qubit gate to precision $\varepsilon$ with an $\{H,T\}$-circuit exist~\cite{RoSelinger,Kliuchnikoff}.
 The number of $T$ gates  in the resulting circuits has scaling close to the
 information-theoretic lower bound of $3\log_2(1/\varepsilon)$
 for $z$-rotations.
Recently, Paetznick and Svore~\cite{PS} showed that by using non-deterministic circuits for decomposition, called Repeat-Until-Success (RUS) circuits, the number of $T$ gates can be further reduced by a factor of $2.5$ on average for axial rotations, and by a larger factor for non-axial rotations.
%
%
This implied that synthesis into RUS circuits can lead to shorter  circuits with {\em expected cost} below the lower bound achieved by a purely unitary circuit design.

We note that the use of measurement to improve the computational power of  unitary circuits is not entirely new.
Research on measurement-based computation \cite{AncillaDriven,AncillaDriven2,MeasCalc} suggests that the use of quantum measurement allows circuits additional computational power at a lower cost in circuit resources.
The use of measurement in the context of decomposition appears in methods of \cite{Jones2012, JonesToff, DCS}, where measurement is used to teleport a quantum state to a target qubit.
In \cite{WK:2013} it was used to obtain efficient circuits for single-qubit unitary operations by trading target precision against accuracy.
Additional discussion of the role of measurement in circuit synthesis appears in \cite{PS} and references therein.

The Paetznick-Svore RUS synthesis algorithm~\cite{PS} is an optimized exhaustive search with exponential classical runtime and limited practicality for a wide range of precisions $\varepsilon$.
In this work, we develop an efficient algorithm to synthesize RUS circuits for approximating a given single-qubit unitary.
Our algorithm runs in probabilistically polynomial classical runtime for any desired precision $\varepsilon$.
Our RUS circuits are composed of Clifford+$T$ gates, a single ancilla qubit, and measurement.
We show that the expected number of  $T$ gates performed upon success scales roughly as $1.15\log_2(1/\varepsilon)$ for axial rotations, a factor of $2.5$ less than the ancilla-free theoretical lower bound and state-of-the-art ancilla-free, unitary methods \cite{RoSelinger,Kliuchnikoff}.
\emph{RUS Circuits.}
The general layout of the Repeat-Until-Success (RUS) circuit protocol is shown in Figure \ref{fig:rus-circuit} \cite{PS}.
Consider a unitary operation $U$ acting on $n+m$ qubits, of which $n$ are target qubits and $m$ are ancillary qubits.
Consider a measurement of the ancilla qubits, such that one measurement outcome is labeled ``success" and all other measurement outcomes are labeled ``failure".
Let the probability of the ``success"  outcome be $p$ and the corresponding unitary applied to the target qubits upon measurement be $V$.
Let $C(U)$ be the cost of the circuit that performs $U$.
We assume for simplicity that any operator $W_i$ performed on target qubits upon a ``failure" measurement is unitary and that each $W_i^{-1}$ can be implemented by a circuit with fixed cost $C(W)$.

\begin{figure}[tb]
\centering
\includegraphics[width=3.3in]{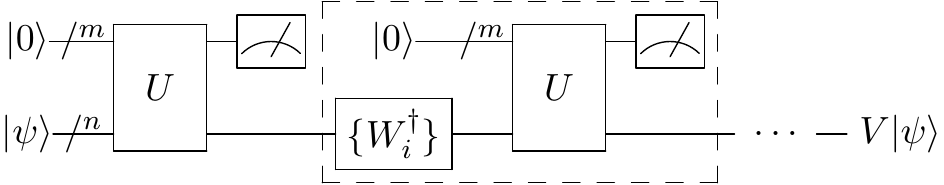}
\caption{\label{fig:rus-circuit}
RUS design circuit to implement unitary $V$.
}
\end{figure}

In the RUS protocol, the circuit in the dashed box is repeated on the $(n+m$)-qubit state until the ``success" measurement is observed.
Each time a ``failure" measurement is observed, an appropriate operator $W^{-1}$ is applied in order to revert the state of the target qubits to their original input state $\ket{\psi}$.
The number of repetitions of the circuit is finite with probability $1$.
%
It is easy to see that the statistical expectation of the overall cost to observe success is:
\begin{equation} \label{expected:cost}
E[C(V)]=\left(C(U) + C(W) \, (1-p)\right)/p.
\end{equation}


We refer to the circuit implementing the unitary operation $U$ as the \emph{RUS design circuit} and its cost $C(U)$ as the \emph{RUS design cost}.
Its expected cost is given by Eq~(\ref{expected:cost}).
The cost of a circuit is measured in the number of $T$ gates.
We choose this cost since fault-tolerant implementations of $T$ gates typically require one to two orders of magnitude more resources than a fault-tolerant Clifford gate \cite{Bravyi2004,MeierEtAl,BravyiHaah}.
Counting $T$ gates has the effect of reducing Eq~(\ref{expected:cost}) to $E[C(V)] = C(U)/p$.

\begin{figure}[bt]
\begin{tabular}{@{\!}c@{\;}c}
\includegraphics[width=1.70in]{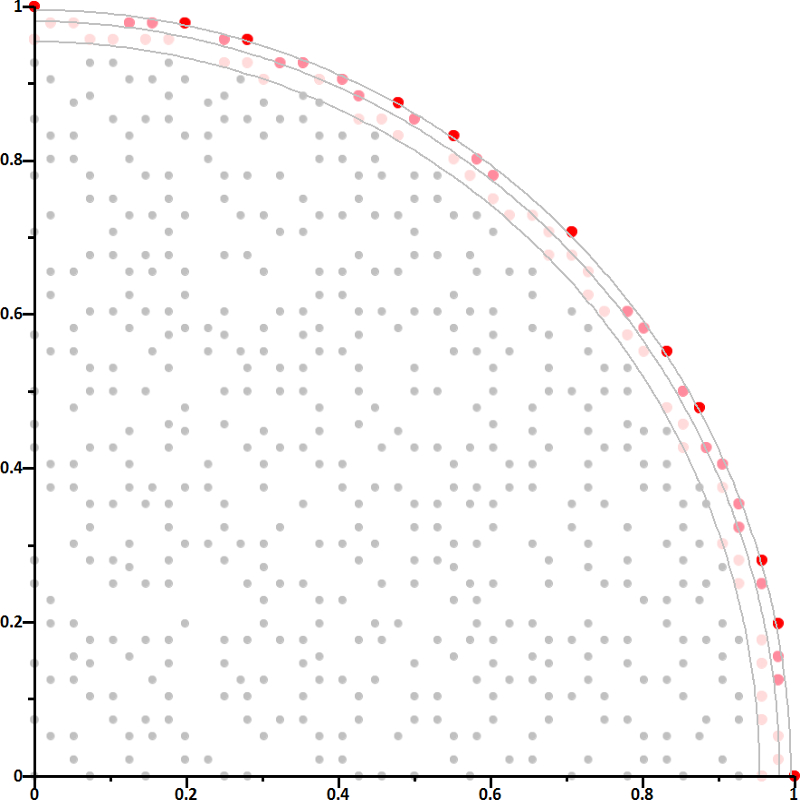} &
\includegraphics[width=1.70in]{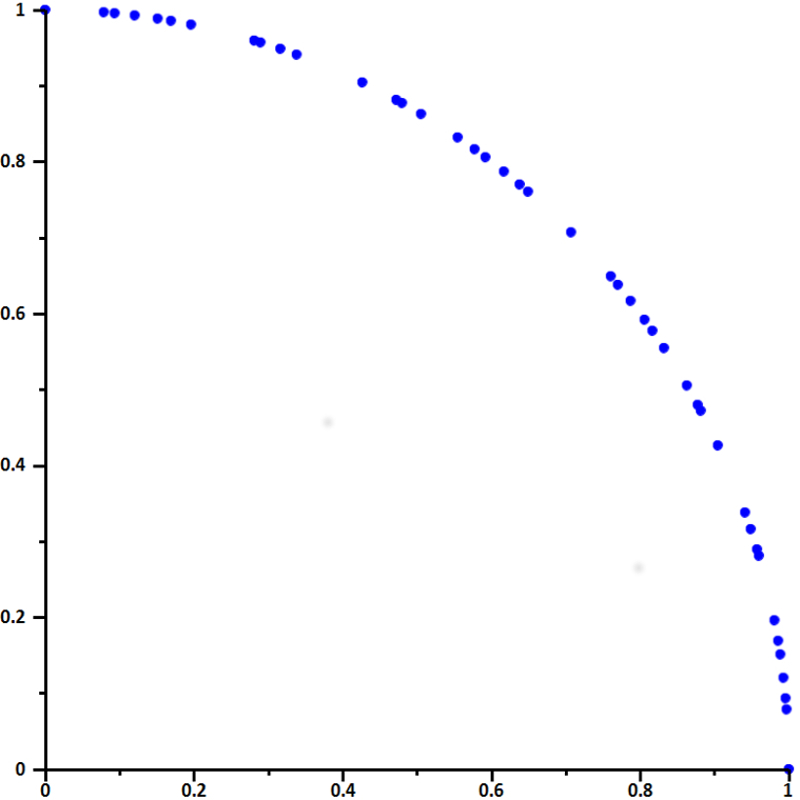}\\
(a) & (b)
\end{tabular}
\caption{\label{fig:density} Comparing approximations of $z$-rotations by (a) unitary $\langle H,T\rangle$ circuits of $T$-depth at most $8$ and (b) RUS protocols with a comparable expected $T$-depth of at most $7.5$.
With respect to the metric $d(V,V')$,
the $144$ $z$-rotations shown in blue in (b) exhibit distances between nearest neighbors of at most $\varepsilon_{max}=0.0676$. In (a) only $40$ circuits are within distance $\varepsilon_{max}$ of any $z$-rotation (dark red points separated by the outer arc), illustrating a higher density of RUS protocols.
In the asymptotic limit for $T$-depth this ratio tends to $3$, one of the main findings of the paper.
Further arcs indicate cutoff distances $\varepsilon_{max}=0.1398$ (bright red)
and $0.2139$ (very bright red),
corresponding to RUS protocols with expected $T$-depths of $6.7$
and $6.2$,
respectively.
}

\end{figure}


We refer to the number of $T$ gates in a circuit as the \emph{$T$-count} and the number of time steps containing $T$ gates as the \emph{$T$-depth}.
The optimal $T$-count has been proven to be an invariant of the unitary operation represented by a Clifford+$T$ circuit \cite{BS12,MatsumotoAmano2008,RandomRemarks}.
In particular, the optimal $T$-count is the same across various definitions of canonical and normal forms for Clifford+$T$ circuits.

%
%
%

In this work, we develop our RUS-based decomposition algorithms for axial rotations.
Any non-axial $V$ can be decomposed into axial rotations such that
\cite{IkeAndMike2000}
\begin{equation}
V = e^{i \delta} \, R_z(\alpha) \, H \, R_z(\beta) \, H \, R_z(\gamma),
\end{equation}
for real values $\alpha, \beta, \gamma, \delta$.
As a matter of principle, a two-qubit RUS design for a single-qubit unitary $V$ can be synthesized directly without breaking $V$ into axial rotations.
However, taking practical advantage of such synthesis is currently an open problem.

\emph{Background.}
%
At the heart of Clifford+$T$ synthesis is the algebraic number ring $\mathbb{Z}[\omega]$, where $\omega = e^{i \pi /4}$, also known as the \emph{ring of cyclotomic integers of order eight}.
The ring has an \emph{integer basis} of four elements, with the most obvious basis being $\{ \omega^3 , \omega^2 , \omega, 1 \}$ \cite{KMM12}.
It consists of all numbers of the form $ a \, \omega^3 + b \, \omega^2 + c \, \omega + d$, where $a,b,c,d$ are arbitrary integers.

It was shown in \cite{KMM12} that a unitary operation $V$ on $n$ qubits is representable exactly by a Clifford+$T$ circuit if and only if it is of the form $V = 1/\sqrt{2}^k M$, where $M$ is a matrix with elements from $\mathbb{Z}[\omega]$ and $k$ is some non-negative integer.
To satisfy the unitary condition, we require $M \, M^{\dagger} = 2^k \, \onemat_{2^n}$.
Moreover, it has been shown that a matrix of this form can be represented as an asymptotically optimal Clifford+$T$ circuit using at most two ancilla qubits \cite{GilSel,Kliuchnikoff}, and no ancilla qubits when either the target is a single-qubit unitary or when $\det(1/\sqrt{2}^k M)=1$~\cite{GilSel}.
We will employ methods of \cite{GilSel,Kliuchnikoff} in our algorithm below.

%

In Figure \ref{fig:density} we illustrate one of the advantages of RUS designs, namely their ability to approximate a given target transformation better than a unitary circuit approximation. Shown in both figures are matrix elements of a unitary $V=\left[\begin{smallmatrix} x &-y^* \\ y & x^*\end{smallmatrix}\right]$. Shown in (a) are the matrix elements $x$ of all unitary matrices $V$ with coefficients in $\Z[i,1/{\sqrt{2}}]$ such that $x$ lies in the upper quadrant of the unit circle and the norm equation $|y|^2=1- |x|^2$ has a solution in $\Z[i,1/{\sqrt{2}}]$ (red points).
Shown in (b) are the matrix elements $x$ of all unitary matrices $V$ that can arise from an RUS protocol as in Figure \ref{fig:rus-circuit} with one ancilla qubit, where the coefficients of $U$ are in $\Z[i,1/{\sqrt{2}}]$ such that $x$ lies in the upper quadrant of the unit circle. Displayed are elements of the form $x/\sqrt{|x|^2+|y|^2}$, where $x=x_0/{\sqrt{2^\ell}}$, $y=y_0/\sqrt{2^\ell}$, where $x_0$, $y_0 \in \Z[\omega]$ and $\ell=3$ (grey points) and the subset of those $x$ for which the norm equation $|z|^2 = 2^\ell - |x_0|^2-|y_0|^2$ has a solution (blue points).

As seen from the figure, the blue points provide a much denser covering than the red points, and thus far better approximations.
We present an analysis of the density with which cyclotomic rationals are distributed in Appendix \ref{sec:theoretic:bounds}.
The density imposes information-theoretic limits on how much we can reduce the expected $T$-count of our non-deterministic RUS solutions compared to the $T$-count of deterministic, unitary solutions.

\emph{Overview of the Algorithm.}
Our algorithm $\varepsilon$-approximates an axial rotation $R_z(\theta)$ by an RUS circuit over the Clifford+$T$ basis in four stages, shown in Figure 1 in Appendix \ref{app:flow}.
We measure the distance between a target unitary $V$ and its approximation $V'$ via the invariant metric $d(V,V') = \sqrt{ 1 - \left(|tr(V^\dagger V')| / 2 \right)}$, see \cite{Fowler:2011}.

The first stage approximates the phase factor $e^{i \theta}$ with a unimodal cyclotomic rational, i.e., an algebraic number of the form $z^* / z$, where $z \in \mathbb{Z}[\omega]$, by finding an approximate solution of an integer relation problem.
We note that $z$ is defined up to an arbitrary real-valued factor.

The second stage performs several rounds of random modification $z \mapsto (r  z)$, where $ r \in \mathbb{Z}[\sqrt{2}]$, in search of an $r$ such that (a) the \emph{norm equation} $|y|^2 = 2^L - |r  z|^2$ is solvable for $ y \in \mathbb{Z}[\omega], \, L \in \mathbb{Z}$, and (b) the one-round success probability $|r  z|^2/2^L$ is sufficiently close to $1$.

In the third stage, the two-qubit matrix corresponding to the unitary part of the RUS circuit is assembled.
During the fourth stage, a two-qubit RUS circuit that implements the desired $R_z(\theta)$ rotation on success and an easily correctable Clifford gate on failure is synthesized.

\emph{Stage 1: Cyclotomic Rational Approximation.}
%
The phase $e^{i\theta}$ is representable exactly as $z^*/z$ if and only if the expression $a \, (\cos(\theta/2)-\sin(\theta/2)) + b \, \sqrt{2}\, \cos(\theta/2) + c \, (\cos(\theta/2)+\sin(\theta/2)) + d \, \sqrt{2} \, \sin(\theta/2)$ is exactly zero (see Appendix \ref{app:approx} for proof).
By making this expression arbitrarily small, then $|z^*/z - e^{i \theta}|$ will be arbitrarily small.
Let $\theta$ be a real number and $z= a \, \omega^3 + b \, \omega^2 + c \, \omega + d, \, a,b,c,d \in \mathbb{Z}$ be a cyclotomic integer.
Then $|z^*/z - e^{i \theta}| < \varepsilon$  if and only if
$| a \, (\cos(\theta/2)-\sin(\theta/2)) + b \, \sqrt{2}\, \cos(\theta/2)$ $ + c \, (\cos(\theta/2)+\sin(\theta/2)) + d \, \sqrt{2} \, \sin(\theta/2)| < \varepsilon \, |z|$, which can be shown
%
by direct complex expansion of $ i e^{-i \theta/2} \, (z^* - e^{i \theta} z) $.

%
To approximate any phase $e^{i \theta}$ with a cyclotomic rational $z^*/z $, where $z \in \mathbb{Z}[\omega]$,
we customize the PSLQ integer relation algorithm \cite{FergBail,PSLQBertok} which attempts to find an integer relation between $(\cos(\theta/2)-\sin(\theta/2))$, $\sqrt{2}\, \cos(\theta/2)$,  $(\cos(\theta/2)+\sin(\theta/2))$, $\sqrt{2} \, \sin(\theta/2)$.
It terminates iterative attempts if and only if $|z^*/z - e^{i \theta}| < \varepsilon$
\footnote{If the trace distance metric is used to approximate unitaries and $\varepsilon$ is the trace distance requested, it would suffice to set $\varepsilon = \sqrt{2} \, \varepsilon$ in this context.}.
Upon termination, our customization also outputs the integer relation candidate $\{a,b,c,d\}$ for which the condition has been satisfied. The desired cyclotomic integer is then given by
$
z = a\, \omega^3 \, + b \, \omega^2 + c \, \omega + d.
$
%
%
%
We find empirically (by simulation) that PSLQ performance is very close to optimal with $ |z| < \kappa \, \varepsilon^{-1/4}$, where $\kappa = 3.05 \pm 0.28$.

\emph{Stage 2: Randomized Search.}
Once the desired $z$ is obtained, the next stage is to include $z$ in a unitary 
\begin{equation} \label{eq:z:y:matrix}
\frac{1}{{\sqrt{2}}^L} \, \left[\begin{array}{cc}
              z & y \\
              -y^* & z^*
            \end{array}\right],
\end{equation}
where $y \in \mathbb{Z}[\omega]$ and $L \in \mathbb{Z}$.
We would like $|z|^2/2^L$ to be reasonably large since this value equals the one-round success probability of the RUS circuit.
Unfortunately, the majority of $z$ values do not allow for this.
%
%
To create a unitary of the form (\ref{eq:z:y:matrix}), we seek a $y$ that satisfies the normalization condition $(|y|^2+|z|^2)/2^L = 1$ , or equivalently $|y|^2 = 2^L - |z|^2$.
It is known \cite{KMM12,Selinger} that $|z|^2$ belongs to the real-valued ring $\mathbb{Z}[\sqrt{2}]$ and thus so does $2^L - |z|^2$.

\begin{figure}[t]
\begin{algorithmic}[1]
\Require $z \in \mathbb{Z}[\omega]$ , size factor $sz$
\Comment hyperparameter $\delta$
\Procedure{RAND-NORMALIZATION-1}{$z, sz$}
\State $L_1 \gets \lceil \log_2(|z|^2) \rceil , cnt \gets 0, Y\gets None, Tc \gets \infty $
\State $S_{\delta} \gets \{a+b\, \sqrt{2}, a,b \in \mathbb{Z},  |a\pm b\,\sqrt{2}| < 2^{\delta\,L_1/2}\}$
\While{ $(cnt++) \leq sz \, \delta \, {L_1}^2$ }
  \State Sample $r$ without replacement from $S_{\delta}$
  \State $L_r \gets \lceil \log_2(|r\, z|^2) \rceil$
  \If {$|y|^2 = 2^{L_r} - |r\, z|^2 $ is easily solvable}
    \State $p \gets |r\, z|^2/2^{L_r}$
    \State $tc \gets Tcount\left[\frac{1}{{\sqrt{2}}^{L_r}} \, \left[\begin{array}{cc}
              z & y \\
              -y^* & z^*
            \end{array}\right]\right]/p$
    \If {$tc < Tc$}
        \State $Tc \gets tc, Y \gets \{ r, y\}$
    \EndIf
  \EndIf
\EndWhile
\EndProcedure
\Ensure $Y$ \Comment the best norm equation solution
\end{algorithmic}
\caption{\label{fig:rand:normalization:one}Randomized normalization algorithm. }
\end{figure}

Given an arbitrary $\xi \in \mathbb{Z}[\sqrt{2}]$, the identity
\begin{equation}\label{eq:norm:equation}
|y|^2 = \xi,
\end{equation}
considered as an equation for an unknown $y \in \mathbb{Z}[\omega]$, is called a \emph{norm equation in} $\mathbb{Z}[\omega]$.
%
Deciding whether a given norm equation is solvable and finding a solution is in general at least as hard as performing factorization of an arbitrary integer.
For our algorithm to be efficient, we need to find norm equations that are easy to solve.
%
%
A necessary condition for easy solvability is that in (\ref{eq:z:y:matrix}), $|z|^2 \leq 2^L$ and $|z^{\bullet}|^2 \leq 2^L$, where $(\cdot)^{\bullet}: \mathbb{Z}[\omega] \rightarrow \mathbb{Z}[\omega]$ extends the $\omega \mapsto (-\omega)$ map.

Our strategy generalizes that of \cite{Selinger}.
We consider a fixed $z \in \mathbb{Z}[\omega]$.
We can replace $z$ in $z^*/z$ by $r z$, where $r \in \mathbb{Z}[\sqrt{2}]$ is arbitrary, without changing the fraction.
For a randomly picked  $r \in \mathbb{Z}[\sqrt{2}]$, we set $L_r = \lceil \log_2(|r  z|^2) \rceil$.
In designing our randomized search, we consider:
(1) the $T$-count of an RUS circuit implementing $(r z)^*/(r z)$ can be made smaller than $ 2 L_r + const$;
(2) its corresponding one-round success probability is given by $p(r) = |r z|^2/2^{L_r}$.

In particular, we want $L_r$ close to its lower bound $L_1 =  \lceil \log_2(|z|^2) \rceil$.
For some small $\delta >0$, we constrain $L_r\leq (1+\delta)L_1$, which implies $r^2 \leq 2^{\delta  L_1}$.
Moreover, we also require $(r^{\bullet})^2 \leq 2^{\delta  L_1}$.
Thus $r$ is sampled from
$S_{\delta}= \{a+b\,\sqrt{2}|\, a,b \in \mathbb{Z} , |a\pm b \, \sqrt{2}|\leq 2^{\delta  L_1/2}\}$.
The cardinality $card(S_{\delta})$ is approximately equal to $2^{1/2 + \delta  L_1}$, corresponding to the area of $\{|a\pm b \, \sqrt{2}|\leq 2^{\delta  L_1/2}\}$ in the $(a,b)$-plane.

While $card(S_{\delta})$ is $O(1/\varepsilon^{\delta})$ and thus exponential in ${\log}_2(1/\varepsilon)$, under a certain working conjecture it suffices to use polylogarithmically many random values of $r$.
We conjecture (and have supporting empirical evidence) that for large enough $\delta  L_1$ there are $\Omega\left(2^{\delta L_1}/(\delta  L_1)\right)$ values of $r \in S_{\delta}$ for which the norm equation $|y|^2 = 2^{L_r}-|r  z|^2$ is easily solvable, and in particular for large enough $k$ there are $\Omega(2^{\delta L_1}/(k\,\delta L_1))$  values for which the equation is solvable and $p(r) > 1-1/k$.

Setting $k=L_1$, we infer from the conjecture that a sample of a size in $O(\delta L_1^2)$ should contain at least one value of $r$ such that the equation $|y|^2 = 2^{L_r}-|r  z|^2$
is easily solvable and $p(r)> 1-1/L_1$. For such $r$ the expected average cost of a RUS circuit that  implements $(r z)^*/(rz)$ is less than
$(2 \,(1+\delta)L_1+const)/(1-1/L_1)$. The latter converges in the asymptotic limit to $2\,(1+\delta)L_1+c_0$, where $c_0$ is a constant.

These observations lead to an algorithm for Stage 2 that randomly samples for the best estimate of the expected average $T$-count of the RUS circuit, as shown in Fig.~\ref{fig:rand:normalization:one}.
It takes the overhead value $\delta$
and sample size factor $sz$ as hyperparameters.
The $Tcount$ function computes the minimal $T$-count of a Clifford+$T$ decomposition of a unitary (without necessarily performing such decomposition) and can be efficiently computed using methods in \cite{BS12,GossetEtAl}.
An alternate algorithm that takes
a minimum success probability as input is described in Appendix \ref{app:norm}.

\emph{Stage 3: RUS Unitary Design.}
When the randomized normalization algorithm succeeds for a given $z$, we can construct a single-qubit unitary $V$ of the form~(\ref{eq:z:y:matrix}), where $y, z \in \mathbb{Z}[\omega], \, L \in \mathbb{Z}$, and $|z|^2/2^L > 1/2$ which maps to the probability of success of the RUS circuit.
%
The unitary $V$ can be decomposed exactly into an optimal ancilla-free Clifford+$T$ circuit using methods in \cite{KMM12}.

The algorithm in Fig.~\ref{fig:single:qubit:design} outputs the unitary $V$.
It calls the randomized normalization algorithm and is designed to combat its infrequent failure.
A failure can only happen if we never encounter an easily solvable norm equation for any random sample.
%
Every iteration of the precision $\epsilon \gets \epsilon/2$ in the while loop of the algorithm adds $sz$ more candidate norm equations to the search space.
For large enough $sz$, the probability of never encountering a solvable norm equation decreases exponentially with the number of iterations.

\begin{figure}[bt]
\begin{algorithmic}[1]
\Require angle $\theta$ , target precision $\varepsilon$, size factor  $sz$
\Comment hyperparameter $\delta$
\Procedure{SINGLE-QUBIT-DESIGN}{$\theta, \varepsilon, sz$}
\State $ret \gets None, \epsilon \gets 2 \, \varepsilon $
\While{ $ret = None$}
  \State $\epsilon \gets \epsilon/2$
  \State Compute $z \in \mathbb{Z}[\omega]$  s.t. $|z^*/z - e^{i \theta}|<\epsilon$
  \Comment Lemma 2
  \State $Y \gets \mbox{RAND-NORMALIZATION-1}(z, sz)$
  \If {$Y \neq None$}
    \State $r \gets first(Y); y \gets last(Y)$
    \State $L \gets \log_2(|r z|^2 + |y|^2)$
    \State $ret \gets \frac{1}{{\sqrt{2}}^{L}} \, \left[\begin{array}{cc}
              r z & y \\
              -y^* & r z^*
            \end{array}\right]$
  \EndIf
\EndWhile
\EndProcedure
\Ensure $ret$ \Comment requisite unitary
\end{algorithmic}
\caption{\label{fig:single:qubit:design}Algorithm to design the unitary $V$. }
\end{figure}

\begin{figure}[bt]
\includegraphics[width=3.5in]{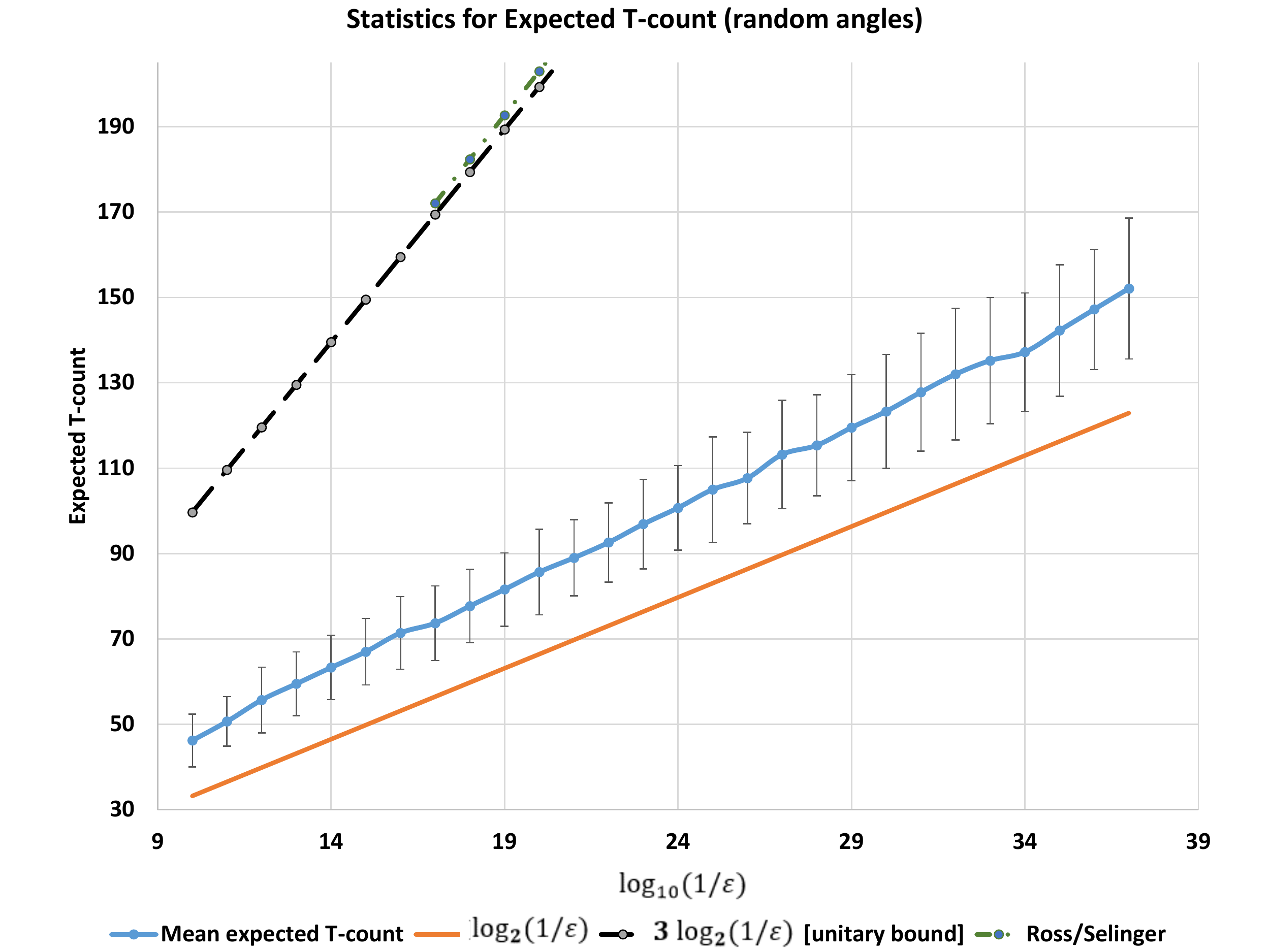}
\caption{\label{fig:random:angles} Precision $\varepsilon$ versus mean expected $T$-count for the set of random angles. } 
\end{figure}

\emph{Stage 4: RUS Circuit Synthesis.}
From $V$ we construct a two-qubit unitary $U$ such that
$U =  \left[\begin{smallmatrix}
              V & 0 \\
              0 & V^{\dagger}
            \end{smallmatrix}\right].
$
We synthesize $U$ as a two-qubit Repeat-Until-Success circuit on one ancilla qubit and one target qubit, such that the circuit applies a $z$-rotation of $z^*/z$
to the target qubit on ``success" and the Pauli-$Z$ operation on ``failure".


The unitary $U$ can be realized as a two-qubit Clifford+$T$ circuit such that
%
the $T$-count is the same or up to $9$ gates higher than the $T$-count of the optimal single-qubit Clifford+$T$ circuit for unitary $V$ (see Appendix \ref{app:synthesis} for the proof).
%
%
The intuition is that given two single-qubit Clifford+$T$ circuits with the same $T$-count, one circuit can be manufactured from the other by insertion and deletion of Pauli gates,
plus the addition of at most two non-Pauli Clifford gates which can result in a small potential $T$-count increase.
%

The pseudocode for the two-qubit RUS circuit synthesis algorithm is given in Fig.~3 in Appendix \ref{app:synthesis}. Moreover, we find that half of the Pauli gates in the RUS circuit can be eliminated using a set of signature-preserving rewrite rules, described in Appendix \ref{sec:circuit:optimize}.
We present an example of applying our RUS synthesis algorithm to the rotation $R_z(\pi/64)$ in Appendix \ref{sec:example}.

\emph{Simulation Results.}
We evaluate the performance of our algorithm on a set of $1000$ angles randomly drawn from the interval $(0,\pi/2)$ at $25$ target precisions $\varepsilon \in \{10^{-11}, \ldots, 10^{-35}\}$.
%
%
Figure \ref{fig:random:angles} plots the precision $\varepsilon$ versus the mean (and standard deviation) expected $T$-count across the RUS circuits generated for the set of $1000$ random angles.
The regression formula for the mean expected $T$-count is $3.817 \, \log_{10}(1/\varepsilon) + 9.2 = 1.149 \, \log_{2}(1/\varepsilon) + 9.2$.
We also plot the mean $T$-count achievable by the methods of \cite{RoSelinger} for the given precisions.
We find similar results on the synthesis of Fourier angles, as detailed in Appendix \ref{subsec:runtime}.

We have also developed generalizations of RUS constructions to a broader set of targets including all unitary operations representable over the field of cyclotomic rationals.
Our generalized designs allow tight control over the $T$-depth and are presented in Appendix \ref{sec:general}.

\emph{Conclusions.}
We have presented an efficient algorithm to synthesize an arbitrary single-qubit gate into a Repeat-Until-Success circuit.
We find that the leading term for the expected $T$-count is given by $c \,\log_2(1/\varepsilon)$, where $c$ is approximately $1.15$ for axial rotations.
On average, our algorithm achieves a factor of $2.5$ improvement over the theoretical lower bound for ancilla-free unitary Clifford+$T$ decomposition.
In turn, it significantly reduces the resources required to implement quantum algorithms on a device.

Future work will extend information-theoretic lower bounds for the expected cost of the RUS circuits to the generalized RUS designs.
We plan to characterize generalized RUS designs and evaluate their expected $T$-count as compared to purely unitary Clifford+$T$ circuits.
Finally, we plan to develop compilation algorithms for implementing generalized RUS designs and to characterize the relationship between the number of ancillas and the properties of the RUS design.

\emph{Acknowledgments.}
We thank Gerry Myerson for pointing us to Wolfgang Schmidt's book \cite{WSchmidt}.
We thank the QuArC team for discussing early versions of this work.


%

\appendix

\section{Information-Theoretic Bounds} \label{sec:theoretic:bounds}

A relatively simple analysis of the density with which cyclotomic rationals are distributed imposes information-theoretic limits on how much we can reduce the expected $T$-count of our non-deterministic RUS solutions compared to the $T$-count of deterministic, unitary solutions.

Let us assume, temporarily, that for $z \in \mathbb{Z}[\omega]$ and $L = \lceil \log_2(|z|^2) \rceil$, the norm equation $|y|^2 = 2^L - |z]^2$ is solvable.
By definition of $L$, $|z|^2 \leq 2^L$.
We know that the optimal $T$-count of a single-qubit unitary circuit implementing a matrix of the form of Eq~(\ref{eq:z:y:matrix}) is $t = 2\,L$ or $t = 2\,L -2$.

\begin{equation} \label{eq:z:y:matrix}
\frac{1}{{\sqrt{2}}^L} \, \left[\begin{array}{cc}
              z & y \\
              -y^* & z^*
            \end{array}\right],
\end{equation}

In either case we note that $|z|^2 = O(2^{t/2})$ and $|z|^4 = O(2^t)$.
We also note that given an upper bound $b$ on the absolute value of cyclotomic integer, there are no more than $O(b^4)$ cyclotomic integers under this bound.
Thus we conclude that there are no more than $O(2^t)$ cyclotomic integers $z$ for which the matrix of the form of Eq~(\ref{eq:z:y:matrix}) may exist and be implemented at $T$-count $t$ or less.

It follows that there are at most $O(2^t)$ unimodular cyclotomic rationals on the unit circumference for which our RUS circuit can be built with design cost of $T$-count $=t$ or less.
Therefore, there exists a constant $K$ such that for $\varepsilon < K \times 2^{-t}$ there is an arc of the unit circumference of length $2 \, \varepsilon$ that does not contain any such cyclotomic rational.
If $\theta_*$ is the angle in the center of such an arc, then the rotation $R_z(\theta_*)$ cannot be implemented by any of our RUS circuits with design cost of $T$-count $=t$ or less.

Conversely, $\varepsilon \geq  K \times 2^{-t}$ is the necessary condition for any axial rotation to be implementable by one of our RUS circuits with design cost of $T$-count $=t$ or less.
This necessary condition is equivalent to
\begin{equation} \label{eq:RUS:lower:bound}
t \geq \log_2(1/\varepsilon) + \log_2(K),
\end{equation}
which is a specific lower bound on the design cost given by the $T$-count of our solution.

Note that in general the integer relations algorithms, and the PSLQ algorithm in particular as discussed in the ``Stage 1: Cyclotomic Rational Approximation" section of this work, provide an opportunity to approach this lower bound.
Indeed, $|z|^4 = O(2^t)$ in the previous discussion suggests $2^t = O(|z|^4) \leq O((\varepsilon^{-1/4})^4) = O(1/\varepsilon)$.
As per the discussion in the "Stage 2: Randomized Search", we do not expect this lower bound to be reachable in practice, because the one-round success rate is always strictly less than one.
Analysis of our particular method only ensures that for any fixed $\delta>0$ there is a precision threshold $\varepsilon_{\delta}$ such that the expected $T$-count is no greater than
$(1+\delta) \log_2(1/\varepsilon) + constant$ can be achieved at precisions $\varepsilon < \varepsilon_{\delta}$. For small values of $\delta$ the  $\varepsilon_{\delta}$ is likely to be impractically tight.

An alternative to explore would be a cost/precision equation of the form $(1+\delta)\log_2(1/\varepsilon) + O(\log(\log(1/\varepsilon)))$  (in the spirit of \cite{RoSelinger}).

The derivation of the above lower bound is specific to our RUS design. Deriving a uniform lower bound under more general assumptions would be a worthwhile problem for future research.

\section{Overview of the Main Algorithm}
\label{app:flow}

The stages of the RUS compilation algorithm are given in Figure \ref{fig:intro:algorithm:flow}.
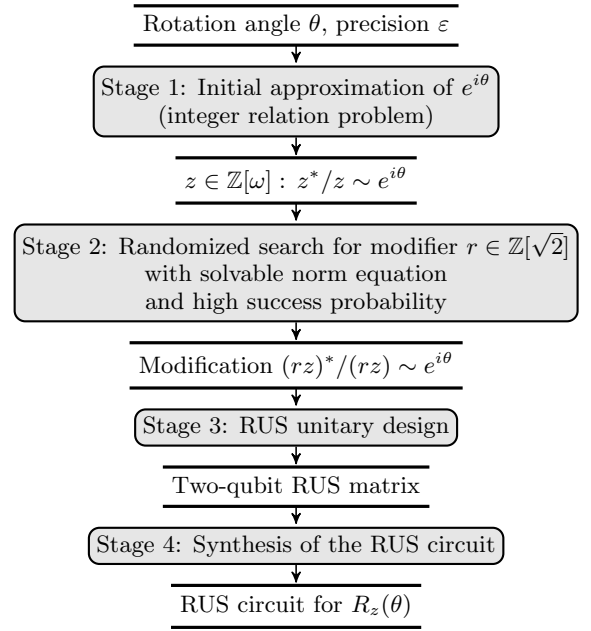
\begin{figure}[hbt]
  \centering
  \makeatletter
\pgfdeclareshape{datastore}{
  \inheritsavedanchors[from=rectangle]
  \inheritanchorborder[from=rectangle]
  \inheritanchor[from=rectangle]{center}
  \inheritanchor[from=rectangle]{base}
  \inheritanchor[from=rectangle]{north}
  \inheritanchor[from=rectangle]{north east}
  \inheritanchor[from=rectangle]{east}
  \inheritanchor[from=rectangle]{south east}
  \inheritanchor[from=rectangle]{south}
  \inheritanchor[from=rectangle]{south west}
  \inheritanchor[from=rectangle]{west}
  \inheritanchor[from=rectangle]{north west}
  \backgroundpath{
    \southwest \pgf@xa=\pgf@x \pgf@ya=\pgf@y
    \northeast \pgf@xb=\pgf@x \pgf@yb=\pgf@y
    \pgfpathmoveto{\pgfpoint{\pgf@xa}{\pgf@ya}}
    \pgfpathlineto{\pgfpoint{\pgf@xb}{\pgf@ya}}
    \pgfpathmoveto{\pgfpoint{\pgf@xa}{\pgf@yb}}
    \pgfpathlineto{\pgfpoint{\pgf@xb}{\pgf@yb}}
 }
}
\makeatother
\usetikzlibrary{arrows}
\begin{tikzpicture}[
  every matrix/.style={ampersand replacement=\&,column sep=1cm,row sep=0.25cm},
  sink/.style={draw,thick,rounded corners,fill=gray!20,inner sep=.1cm},
  datastore/.style={draw,very thick,shape=datastore,inner sep=.1cm},
  to/.style={->,>=stealth',shorten >=1pt,semithick},
  every node/.style={align=center}]
  \matrix{
    \& \node[datastore] (ang) {Rotation angle $\theta$, precision $\varepsilon$}; \& \\
    \& \node[sink] (appr) {Stage 1: Initial approximation of $e^{i \theta}$\\
                            (integer relation problem)}; \& \\
    \& \node[datastore] (inia) {$z \in \mathbb{Z}[\omega] : \, z^*/z \sim e^{i \theta}$}; \& \\
 	\& \node[sink] (rsea) {Stage 2: Randomized search for modifier $r \in \mathbb{Z}[\sqrt{2}]$\\
                            with solvable norm equation\\
                            and high success probability}; \& \\
	\& \node[datastore] (moda) {Modification $(r z)^*/(r z) \sim e^{i \theta}$}; \& \\
    \& \node[sink] (design) {Stage 3: RUS unitary design}; \& \\
	\& \node[datastore] (matrix) {Two-qubit RUS matrix}; \& \\
 	\& \node[sink] (synt) {Stage 4: Synthesis of the RUS circuit}; \& \\
	\& \node[datastore] (out) {RUS circuit for $R_z(\theta)$}; \& \\
  };
	\draw[to] (ang) --(appr);
	\draw[to] (appr) --(inia);
	\draw[to] (inia) --(rsea);
	\draw[to] (rsea) --(moda);
    \draw[to] (moda) --(design);
	\draw[to] (design) --(matrix);
	\draw[to] (matrix) --(synt);
	\draw[to] (synt) --(out);
\end{tikzpicture}
 \caption[Compilation algorithm]{Overview of the compilation algorithm.
}
\label{fig:intro:algorithm:flow}
\end{figure}

\section{Details on Cyclotomic Rational Approximation}\label{app:approx}

\begin{lem} \label{lem:integer:relation}
Let $\theta$ be a real number and $z= a \, \omega^3 + b \, \omega^2 + c \, \omega + d, \, a,b,c,d \in \mathbb{Z}$ be a cyclotomic integer.
Then $|z^*/z - e^{i \theta}| < \varepsilon$  if and only if
$| a \, (\cos(\theta/2)-\sin(\theta/2)) + b \, \sqrt{2}\, \cos(\theta/2)$ $ + c \, (\cos(\theta/2)+\sin(\theta/2)) + d \, \sqrt{2} \, \sin(\theta/2)| < \varepsilon \, |z|$.
\end{lem}

\begin{proof}
By direct complex expansion of $ i \, e^{-i \theta/2} \, (z^* - e^{i \theta}\, z) $.
\end{proof}

Lemma \ref{lem:integer:relation} implies that the phase $e^{i \theta}$ is representable exactly as $z^*/z$ if and only if the expression $a \, (\cos(\theta/2)-\sin(\theta/2)) + b \, \sqrt{2}\, \cos(\theta/2) + c \, (\cos(\theta/2)+\sin(\theta/2)) + d \, \sqrt{2} \, \sin(\theta/2)$ is exactly zero.
It also implies that by making this expression arbitrarily small, then $|z^*/z - e^{i \theta}|$ will be arbitrarily small.

We make the following observation regarding a bound on the size of the cyclotomic integer:
\begin{observ} \label{observe:size}
In the context of Lemma \ref{lem:integer:relation}, for any angle $\theta$ such that $0<|\theta|<\pi/2$, there exists a set of parameterized cyclotomic integers $z_{\varepsilon}$ such that $z^{*}_{\varepsilon}/z_{\varepsilon}$  is an $\varepsilon$-approximation of $e^{i \theta}$ and  $|z_{\varepsilon}|=O(\varepsilon^{-1/4})$.
\end{observ}

\begin{proof}
The proof follows from a more general theorem regarding the quality of multivariate Diophantine approximations (c.f.,~\cite{WSchmidt}, Theorem 1C):
For any real numbers $x_1,...,x_n$ and $0 < \varepsilon < 1$ there exist integers $q_1,...,q_n,p$ such that $|q_1 \, x_1 + \cdots + q_n \, x_n - p | < \varepsilon$ and $\max(|q_1|,\ldots,|q_n|))< \varepsilon^{-1/n}$.

We apply this theorem to our case for $n=3$.
Noting that $1 < \sqrt{2}\, \cos(\theta/2) < \sqrt{2}$ , introduce
$x_1=(\cos(\theta/2)-\sin(\theta/2))/(\sqrt{2}\, \cos(\theta/2))$, $x_2=(\cos(\theta/2)+\sin(\theta/2))/(\sqrt{2}\, \cos(\theta/2))$, and $x_3=\sin(\theta/2)/\cos(\theta/2)$.

For $0<\delta'<1$, there exist $a,b,c,d \in \mathbb{Z}$ such that $|a\,x_1+c\,x_2+d\,x_3+b|<\delta'$ and $\max(|a|,|c|,|d|)<\delta'^{-1/3}$. However, in our context  $|x_1|<1/\sqrt{2},|x_2|<1/\sqrt{2},|x_3|<1$ thus $|b| < (2+\sqrt{2}) \, \delta'^{-1/3}$.

Select $\delta = \sqrt{2} \delta' $ and corresponding $a,b,c,d$ as suggested above.
Then
\begin{eqnarray*}
|a \, (\cos(\theta/2)-\sin(\theta/2)) &+& b \, \sqrt{2}\, \cos(\theta/2) + \\
c \, (\cos(\theta/2)+\sin(\theta/2)) &+& d \, \sqrt{2} \, \sin(\theta/2)| < \delta,
\end{eqnarray*}
$\max(|a|,|c|,|d|)<\sqrt{2} \, \delta^{-1/3}$, and $\,|b| < (2+2 \sqrt{2}) \, \delta^{-1/3}$.

We observe that for $z = a\, \omega^3 \, + b \, \omega^2 + c \, \omega + d$, we have $|z|\leq (|a|+|b|+|c|+|d|) < (2+5 \sqrt{2}) \, \delta^{-1/3} $.

Suppose that integers $a,b,c,d$ are the smallest in magnitude for which the above inequality is satisfied for $\delta = \varepsilon\, |z|$.
Then we have $|z|\leq (2+5 \sqrt{2}) \, \varepsilon^{-1/3} |z|^{-1/3}$ from which we obtain that $|z| < (2+5 \sqrt{2})^{3/4} \, \varepsilon^{-1/4}$.

\end{proof}

The theorem in the original PSLQ paper \cite{FergBail} states that if exact integer relations exist and if $m$ is the minimum norm of such a relation, then the norm of the relation found by the PSLQ algorithm on termination is in $O(m)$.
The theorem and its proof can be generalized to cover approximate integer relations sought by our customization.
In the generalized form, the theorem states that if $m_{\varepsilon}$ is the minimum norm of an approximate integer relation corresponding to the requested precision $\varepsilon$, then the norm of a relation actually found by the customized algorithm is in $O(m_{\varepsilon})$ and thus, as per Observation \ref{observe:size}, it is in $O(\varepsilon^{-1/4})$.

\section{Details on the Norm Equation in $\mathbb{Z}[\omega]$}
\label{app:norm}
We recall first that the real-valued ring $\mathbb{Z}[\sqrt{2}]$ is a unique factorization ring.
That is, any of its elements can be factored into a product of prime algebraic integers and at most one unit.
The primary category of right-hand-side values for which Eq~(\ref{eq:norm:equation}) is easily solvable would then be the set of algebraic integer primes.

\begin{equation}\label{eq:norm:equation}
|y|^2 = \xi,
\end{equation}

Specifically, the equation is constructively solvable in the following three situations (c.f.,~\cite{LWashington}):
\begin{enumerate}
\item $\xi = 2 \pm \sqrt{2}$;
\item $\xi = a + b \, \sqrt{2}$, $\xi > 0$  and $p=a^2-2 \, b^2$ is a positive rational prime number with $p = 1 \, \mod \, 8$;
\item $\xi$ is a rational prime number and $\xi \neq -1 \, \mod \, 8$.
\end{enumerate}
We call an algebraic integer prime belonging to one of these three classes a ``good'' prime.

For a composite $\xi$ we consider a \emph{limited} factorization of the right-hand side to preserve efficiency.
To this end, we precompute a set $S_{prime} \subset \mathbb{Z}[\sqrt{2}]$ of small prime elements and consider factorizations of the form:
$\xi =  \xi_1^{a_1} \, \ldots, \xi_r^{a_r} \, \eta$, where $\xi_1, \ldots, \xi_r \in S_{prime}$ and $\eta$ passes a primality test.
Eq~(\ref{eq:norm:equation}) is efficiently solvable if $\eta$ is a good prime and for $i=1, \ldots , r$, $\xi_i$ is a good prime or $a_i$ is even.

\begin{example}
$|y|^2 = \xi = 1270080 + 211680 \, \sqrt{2}$ is efficiently solvable since
$\xi= 2^5 \, 3^3 \, 5 \, 7^2 \, (2+\sqrt{2}) (5- 2\, \sqrt{2})$.
\end{example}

Note $p=5^2 - 2\, 2^2 = 17 = 1 \, \mod \, 8$.
The only ``bad" prime in the above factorization is $7$ but it appears as an even power.

We remark that the cyclotomic integer $z$ coming from the cyclotomic rational approximation of
$e^{i \theta}$ is not unique.
In fact, it is defined up to an arbitrary real-valued factor $r \in \mathbb{Z}[\sqrt{2}]$.
For any such $r$, $(r\, z)^* / (r \, z)$  is identical to $z^*/z$.
However the norm equation $|y|^2 = 2^L - |r \, z|^2$ can and will change quite dramatically.

When drawing $r$ randomly from a subset of $\mathbb{Z}[\sqrt{2}]$ one might try and estimate the chance that the equation $|y|^2 = 2^L - |r \, z|^2$ turns out to be solvable for a random $r$.
This is an example of an open and likely very hard number theory problem.
We will not attempt to solve it here and will instead rely on a conjecture that the ``lucky" values of $r$ are reasonably dense in $\mathbb{Z}[\sqrt{2}]$.
A formal statement of this conjecture is given in the next subsection.

\subsection{Random Search Normalization} \label{subsec:random:search}

The second version of the algorithm (Fig.~\ref{fig:rand:normalization:two}) has the sampling radius as the only hyperparameter and the minimum acceptable one-round success rate as an additional input.
It is designed as a building block for generating ``near-deterministic'' RUS circuits with high (single round) success probability.

External to both versions of the algorithm is the $Tcount$ function that computes the minimal $T$-count of a Clifford+$T$ decomposition of a unitary (without necessarily performing such decomposition).
The existence and efficiency of the $Tcount$ function has been proven recently by several researchers \cite{BS12,GossetEtAl}.
Their methods can be applied to calculate the $T$-count in the algorithms of Fig. \ref{fig:rand:normalization:two}.

\begin{figure}[t]
\begin{algorithmic}[1]
\Require $z \in \mathbb{Z}[\omega]$ , minimum success probability $p_{min}$
\Comment hyperparameter $\delta$
\Procedure{RAND-NORMALIZATION-2}{$z, p_{min}$}
\State $L_1 \gets \lceil \log_2(|z|^2) \rceil , cnt \gets 0, Y\gets None, Tc \gets \infty $
\State $S_{\delta} \gets \{a+b\, \sqrt{2}, a,b \in \mathbb{Z} , |a\pm b\,\sqrt{2}| < 2^{\delta\,L_1/2}\}$
\While{ $(cnt++) \leq card(S_{\delta})$ }
    \State Sample $r$ without replacement from $S_{\delta}$
  \State $L_r \gets \lceil \log_2(|r\, z|^2) \rceil$
  \If {$|y|^2 = 2^{L_r} - |r\, z|^2 $ is easily solvable}
    \State $p \gets |r\, z|^2/2^{L_r}$
    \State $tc \gets Tcount\left[\frac{1}{{\sqrt{2}}^{L_r}} \, \left[\begin{array}{cc}
              z & y \\
              -y^* & z^*
            \end{array}\right]\right]/p$
    \If {$p > p_{min}$ and $tc < Tc$}
        \State $Tc \gets tc, Y \gets \{ r, y\}$
    \EndIf
  \EndIf
\EndWhile
\EndProcedure
\Ensure $Y$ \Comment the best norm equation solution
\end{algorithmic}
\caption{\label{fig:rand:normalization:two}Randomized normalization algorithm, version 2.}
\end{figure}

%


\section{Details on RUS Circuit Synthesis}
\label{app:synthesis}
We begin with the following general theorem and develop the corresponding proof through a sequence of lemmas:
\begin{thm} \label{thm:block:diagonal}
Let $V, W$ be single-qubit unitaries representable by single-qubit ancilla-free Clifford+$T$ circuits with the same minimal $T$-count $t$.
Then the two-qubit unitary $U=\left[\begin{array}{cc}
              V & 0 \\
              0 & W
            \end{array}\right]$
is exactly and constructively representable by a Clifford+$T$ circuit with $T$-count at most $t+9$.
\end{thm}

The intuition is that given two single-qubit Clifford+$T$ circuits with the same $T$-count, one of the circuits can be manufactured from the other by insertion and deletion of Pauli gates, plus the addition of at most two non-Pauli Clifford gates.
The two non-Pauli Clifford gates are the only gates responsible for the potential $T$-count increase, which stems from lifting the $V,W$ pair to the desired two-qubit RUS circuit.

The following three definitions of single-qubit Clifford+$T$ circuits will be used in the proofs.
\begin{defin}
A \emph{$T$ code} is a circuit generated by the  $T \, H$ and $T^{\dagger} \, H$ syllables.
\end{defin}
\begin{observ}
If $c$ is a $T$ code then $H \, c^{-1} \, H$ is a $T$ code.
\end{observ}
\begin{defin}
A \emph{decorated $T$ code} is a circuit generated by syllables of the form $P \, T^{\pm 1} \, Q \, H$, where $P,Q \in \{Id,X,Y,Z\}$.
\end{defin}
\begin{defin}
A decorated $T$ code $c_2$ is called a \emph{decoration} of a $T$ code $c_1$ if $c_1$ can be obtained from $c_2$ by removing all explicit occurrences of Pauli gates.
\end{defin}

The following lemma establishes yet another normal form for single-qubit Clifford+$T$ circuits (see \cite{RandomRemarks} for a short review of useful normal forms).
\begin{lem} \label{lem:t:odd}
Any single-qubit ancilla-free Clifford+$T$ circuit can be constructively rewritten in the form $g_1 \, c \, g_2$ where $c$ is a $T$ code and $g_1,g_2$ are single-qubit Clifford gates.
\end{lem}

To prove Lemma \ref{lem:t:odd} and other subsequent results we require the following relations, which are established by direct computation:
\begin{lem} \label{lem:t:identities}
The following are exact identities:

$X\, T\,H = T^{\dagger} \, H \, Z \, \omega$;

$Y\, T\, H = T^{\dagger} \, H \, Y \, \omega^5$;

$ Z\, T\, H = T\, H\,X$;

$X\, S\, H\, T\, H = S^{\dagger} \, H \, T\, H \, X \, \omega^2$;

$Y\, S\, H\, T\, H = S^{\dagger} \, H \, T^{\dagger}\, H \, Y \, \omega^3$;

$Z\, S\, H\, T\, H = S \, H \, T^{\dagger}\, H \, Z \, \omega$.
\end{lem}

\begin{proof}{Of Lemma \ref{lem:t:odd}.}

In \cite{BS12}, canonical $\langle T,H \rangle$-circuits are defined as $\langle T,H \rangle$-circuits that are representable by an arbitrary composition of the $T \, H$ and $S \, H \, T \, H$ syllables and start with the $T\, H$ syllable.
An algorithm for constructively rewriting any single-qubit ancilla-free Clifford+$T$ circuit to the form $g_1 \, c \, g_2$ where $c$ is a canonical circuit is also given.

Hence to complete the proof of the lemma it is sufficient to show that a canonical circuit can be constructively rewritten to a form $c'\, P \, \omega^m$, where $c'$ is a $T$ code, $P$ is a Pauli gate, and $m$ is an integer.
The proof is by induction on the number of syllables in the canonical circuit.

In case of $0$ or $1$ syllables there is nothing to prove.

Consider a canonical circuit with $k > 1$ syllables.
Such a circuit is by definition equal to either $c \, T\, H$ or to $c \, S \, H \, T \, H$, where $c$ is a canonical circuit with $k-1$ syllables.
By the induction hypothesis, $c$ can be recursively rewritten as $c_1 \, T^{d} \, H\, Q \, \omega^l$, where $c_1$ is a $T$ code (empty when $k=2$), $Q$ is a Pauli gate, $l$ is an integer, and $d=\pm 1$.

If the subject circuit is $c \, T\, H$, we apply one of the first three identities from Lemma \ref{lem:t:identities} to make the induction step.
If the subject circuit is $c \, S \, H \, T \, H$ , we need to run a case distinction on $Q$ and $d$.

We note the following exact rewrites:

$H \, S \, H = S^{\dagger} \, H \, S^{\dagger} \, \omega = T^{-2} \, H \, T^{-2} \, \omega$;

$H \, S^{\dagger} \, H = S \, H \, S \, \omega^{-1} = T^2 \, H \, T^2 \, \omega^{-1}$.

Next we apply one of the bottom three identities from Lemma \ref{lem:t:identities} to each case as needed.

If $Q=Id$, we get
$c_1 \, T^{d} \, H\, S\, H \, T \, H \,  \omega^{l} =  c_1 \, T^{d-2} \, H\, T^{\dagger} \, H \,  \omega^{l+1} $.
If $Q=Id$ and $d=1$ and we are done.
If $Q=Id$ and $d=-1$, we get $c_1 \, T^{-3} \, H\, T^{\dagger} \, H \,  \omega^{l+1} $ which needs to be further rewritten as
$c_1 \, T\, Z \, H\, T^{\dagger} \, H \,  \omega^{l+1} =  $ $c_1 \, T\,  H\, X \, T^{\dagger} \, H \,  \omega^{l+1} = $
$c_1 \, T\,  H\,  T \, X \, H \,  \omega^{l} = $ $c_1 \, T\,  H\,  T \, H \,  Z \, \omega^{l}$.

If $Q=X$, we get
$c_1 \, T^{d} \, H\, S^{\dagger}\, H \, T \, H \, X \, \omega^{l+2} =  $  $c_1 \, T^{d+2} \, H \, T^3 \, H \, X \, \omega^{l+1} = $
$c_1 \, T^{d+2} \, H \, T^{\dagger} Z\,  \, H \, X \, \omega^{l+1} = $ $c_1 \, T^{d+2} \, H \, T^{\dagger} \, H \, Id \, \omega^{l+1}$.
If $Q=X$ and $d=-1$, we are done.
If $Q=X$ and $d=+1$, we get $c_1 \, T^3 \, H \, T^{\dagger} \, H \, Id \, \omega^{l+1}$ which needs to be further rewritten as
$c_1 \, T^{\dagger} \, Z \, H \, T^{\dagger} \, H \, Id \, \omega^{l+1}= $ $c_1 \, T^{\dagger} \, H \, X \, T^{\dagger} \, H \, Id \, \omega^{l+1}= $
$c_1 \, T^{\dagger} \, H \, T \,  X \, H \, Id \, \omega^{l}= $ $c_1 \, T^{\dagger} \, H \, T \,  H \, Z \, \omega^{l}$.

If $Q=Y$, we get $c_1 \, T^{d} \, H\, S^{\dagger} \, H \, T^{\dagger} \, H  \, Y \, \omega^{l+3} = $ $c_1 \, T^{d+2} \,  H \, T \, H  \, Y \, \omega^{l+2}$.
If $Q=Y$ and $d=-1$, we are done.
If  $Q=Y$ and $d=1$, we get $c_1 \, T^{\dagger} \, Z \,  H \, T \, H \,  Y \, \omega^{l+2}$ which needs to be further rewritten as
$c_1 \, T^{\dagger} \,  H \,  X \, T \, H \,  Z \, \omega^{l+2}=$
$c_1 \, T^{\dagger} \,  H \, T^{\dagger} \,  X \,  H \,  Z \, \omega^{l+3}=$
$c_1 \, T^{\dagger} \,  H \, T^{\dagger} \, H \,  Id \, \omega^{l+3}$.

If $Q=Z$, we get $c_1 \, T^{d} \,  H \, S \, H \, T^{\dagger} \, H \, Z \, \omega^{l+1}=$
$c_1 \, T^{d-2} \,  H \, T \, Z \,  H \, Z \, \omega^{l+2} =$
$c_1 \, T^{d-2} \,  H \, T \,   H \, X \, Z \, \omega^{l+2} =$
$c_1 \, T^{d-2} \,  H \, T \,   H \, Y \, \omega^{l}$.
If $Q=Z$ and $d=1$, we are done.
If $Q=Z$ and $d=-1$, we get $c_1 \, T \, Z \,  H \, T \, H \, Y \, \omega^{l}$ which needs to be further rewritten as
$c_1 \, T \,  H \, X \, T \, H \, Y \, \omega^{l}=$
$c_1 \, T \,  H \,  T^{\dagger} \, X \,  H \, Y \, \omega^{l+1}=$
$c_1 \, T \,  H \,  T^{\dagger} \,   H \, X \, \omega^{l-1}$.

This concludes the induction step.

\end{proof}

\begin{lem} \label{lem12:pauli:decoration}
For any $T$ code $c_1$, any single-qubit ancilla-free Clifford+$T$ circuit $c_2$ with the same $T$-count as $c_1$ can be constructively rewritten
as $g_3 \, c'_{1} \, g_4$, where $c'_{1}$ is a decoration of $c_1$ and $g_3, g_4$ are Clifford gates.
\end{lem}

\begin{proof}
As per Lemma \ref{lem:t:odd}, $c_2$ can be constructively rewritten as $g_3 \, c'_2 \, g_2$, where $c'_2$ is a $T$ code.
Clearly $c_1$ and $c'_2$ have the same $T$-count, i.e., the same number $t$ of syllables of the form $T^{\pm 1}\,H$.

We show by induction on $t$ that there is a certain decoration $c'_1$ of $c_1$ and a certain integer $m$ such that $c'_2 = c'_1 \, \omega^m$.
For $t=0$ the statement is trivial.

Suppose the statement has been proven for circuits with fewer than $t$ syllables.
We have $c_1 = T^{d_1} H \, c_4$; $c'_2 = T^{d_2} H \, c_5, d_1,d_2 \in \{-1,+1\}$ , where, by the induction hypothesis, there exists a decoration $c'_4$ of $c_4$ and an integer $s$ such that $c_5 = c'_4 \, \omega^s$ and thus $c'_2 = T^{d_2} H \, c'_4 \, \omega^s$.

If $d_1 = d_2$, the induction step is complete.

If $d_1 = 1$, $d_2 = -1$, then  $c'_2 = c'_1 \,  \omega^{s-1}$ where $c'_1 =  X\, T \, X \, H \, c'_4$. Set $g_4= \omega^{s-1}\, g_2$.

If $d_1= -1$, $d_2 = 1$, then $c'_2 = c'_1 \,  \omega^{s+1}$ where $c'_1 =  X\, T^{\dagger} \, X \, H \, c'_4$. Set $g_4= \omega^{s+1}\, g_2$.

In both cases $c'_1$ is a decoration of $c_1$ which completes the induction step.

\end{proof}

\begin{corol}
Let $c_1$ and $c_2$ be two single-qubit Clifford+$T$ circuits with the same $T$-count $t$. One can constructively find a  $T$ code $c_3$, one of its decorations $c_4$, and Clifford gates $g_1,g_2,g_3,g_4$ such that $c_1 = g_1 \, c_3 \, g_2$ and $c_2 = g_2 \, c_4 \, g_4$.
\end{corol}

Informally, up to Clifford gate wrappers, any two single-qubit Clifford+$T$ circuits are related via $T$-code decoration.
Next we define the $Lift$ procedure to convert a decorated $T$ code into a certain two-qubit circuit.
We use the notation $\Lambda(G)$ to denote a controlled-$G$ unitary.
\begin{defin}
For any Pauli gate $P$, $Lift(P) = \Lambda(P)$.
For any other non-Pauli gate $g$, $Lift(g) = Id \otimes g$.
Given a single-qubit Clifford+$T$ circuit $c= g_1 \ldots g_r$ where $g_i \in \{X,Y,Z,H,T,T^{\dagger}\}$ and $i = 1,\ldots,r$, $Lift(c)= Lift(g_1) \ldots Lift(g_r)$.
\end{defin}

\begin{lem}{(The ``Jack of Daggers" lemma.)}\label{lem:JOD}
If a single-qubit unitary $V$ is represented as $H \, c_1$ where $c_1$ is a $T$ code of $T$-count $t$, then the two-qubit unitary
$J(V) = \left[\begin{array}{cc}
              V & 0 \\
              0 & V^{\dagger}
            \end{array}\right]
$
is constructively represented by a two-qubit Clifford+$T$ circuit of $T$-count either $t$ or $t+1$.
\end{lem}

\begin{proof}
We infer this from the preceding lemmas and propose a two-qubit circuit layout that will be convenient for our subsequent constructions.
We start by constructively decomposing $V^{\dagger}$ into $H \, c_2 \, \omega^k, k \in \mathbb{Z}$, where $c_2$ is a decoration of the circuit $c_1$.
Since $(H \, c_1)^{-1} = (c_1)^{-1} \, H$ is $H$ followed by a $T$ code, the desired decomposition is done by applying the procedure from the proof of Lemma \ref{lem12:pauli:decoration}. (We note that in this special case the only additional Clifford gates that appear are powers of $\omega$.)

It is easy to see that the circuit $Lift(H\, c_2) \, (T^k \otimes Id)$ represents the subject two-qubit unitary.
The $T$-count of $Lift(H\, c_2)$ is $t$ and the $T$-count of $(T^k \otimes Id)$ is either $0$ or $1$.
\end{proof}

Lemma \ref{lem:JOD} is a special case of Thm \ref{thm:block:diagonal}, which we are now ready to prove.
\begin{proof}{Of Thm \ref{thm:block:diagonal}.}

By assumption of the theorem, the upper left unitary $V$ is representable by a single-qubit Clifford+$T$ circuit.
As per Lemma \ref{lem:t:odd}, $V$ can be constructively represented as $g_1 \, c \, g_2$ with Clifford gates $g_1,g_2$ and $T$ code $c$.
As per Lemma \ref{lem12:pauli:decoration}, the lower unitary $W$ can be constructively represented as $g_3 \, g_1 \, c' \, g_2 \, g_4$, where $c'$ is a decoration of $c$ and $g_3,g_4$ are Clifford gates.

Consider  $c''=Lift(c')$.
It is easy to see that $U$ is represented by
$\Lambda(g_3) $ $(Id\otimes g_1) \, c'' \, (Id\otimes g_2) \, \Lambda(g_4)$.
We note that the occurrence of each controlled-Clifford gate increases the $T$-count by at most $5$.
Moreover by applying global phase carefully we can always guarantee that the $T$-count of $\Lambda(g_3)$ is at most $4$.

By inspection of all single-qubit Clifford gates we find that for any Clifford gate $g$, the controlled gate $\Lambda(g)$ can be can be represented as a product of some Clifford gates and exactly one gate of the form $\Lambda(h)$, where $h \in \{Id, \omega^{-1} \, S, H, \omega^{-1} \, S \, H, \omega^{-1} \, H \, S\}$.
$\Lambda(Id)$ has $T$-count $0$.
By known methods (c.f.,~\cite{GilSel}), $\Lambda(\omega^{-1} \, S)$  and $\Lambda(H)$ can be constructed with $T$-count $2$, hence by composition $\Lambda(\omega^{-1} \, H \, S)$  and
$\Lambda(\omega^{-1} \, S\, H)$ can be constructed with $T$-count $4$.

Thus if $\Lambda(g_3)$ can only be constructed with $T$-count $5$, we can replace $g_3$ with $\omega^s \, g_3$ for an appropriate integer $s$  such that
$\Lambda(\omega^s \, g_3)$ can be constructed with $T$-count $4$ (while at the same time replacing $g_4$ with $\omega^{-s} \, g_4$).
\end{proof}

It might be worthwhile noting that the $Lift$ operation, as designed and applied places all the single-qubit gates on the ancillary qubit, while the primary qubit is used solely as control. Lemma \ref{lem:JOD} leads to the algorithm shown in Fig. \ref{fig:parsimonious:circuit:synthesis} to implement RUS designs that are based on decomposing matrices $J(V)$ of the ``Jack of Daggers'' form.

\begin{figure}[t]
\begin{algorithmic}[1]
\Require single-qubit unitary $V$
\Comment representable in Clifford+$T$
\Procedure{RUS-SYNTHESIS}{$V$}
\State represent $V$ as $g_1\, c \, g_2$ where $c$ is T code
\Comment $g1, g2$ are Clifford gates
\State represent $V^{\dagger}$ as $g_3 \, c' \, g_4$ where $c'$ is a decoration of $c$
\Comment $g3, g4$ are Clifford gates
\State $g_5 \gets g_3\, g_1^{\dagger} , g_6 \gets g_2^{\dagger} g_4$
\State $ret \gets \Lambda(g_5) (Id \otimes g_1) Lift(c') (Id \otimes g_2) \Lambda(g_6)$
\EndProcedure
\Ensure $ret$ \Comment RUS design circuit
\end{algorithmic}
\caption{\label{fig:parsimonious:circuit:synthesis}Algorithm to synthesize an RUS from $V$. }
\end{figure}

Since we are optimizing $T$-count only and consider the entanglers to be zero-cost, this approach does not limit generality. Using swap entanglers the single-qubit gates can be moved freely between the primary and the ancillary qubit. In that sense more general circuit designs that employ more complicated single-qubit gate patterns are in fact functionally equivalent to what we do here.

\begin{corol} \label{corol:RUS:synthesis}
Let $V$ be a single-qubit unitary of the form $V= \left[\begin{array}{cc}
              u& v \\
              -v^* & u^*
            \end{array}\right]
$
that is exactly representable by a single-qubit Clifford+$T$ circuit of minimum $T$-count $t$.
Then the rotation
$R=\left[\begin{array}{cc}
              1 & 0 \\
              0 & u^*/u
            \end{array}\right]
$
can be exactly and constructively represented by a two-qubit Repeat-Until-Success circuit with $T$-count at most $t+9$ and single-round success probability $|u|^2$.
\end{corol}

\begin{proof}
Consider the two-qubit operator $U = J(V)=\left[\begin{array}{cc}
              V & 0 \\
              0 & V^{\dagger}
            \end{array}\right].
$
Consider a two-qubit circuit where the first qubit is the target and the second qubit is an ancilla.
The input of the circuit is $|\psi \rangle \otimes |0\rangle$, the operator $U$ is applied to the input followed by measurement of the ancilla qubit.
By direct computation, the probability of measuring $0$ (``success") is $|u|^2$ and the state of the target qubit is equivalent to $R \, |\psi \rangle$ upon measurement.
It is equally straightforward that upon measurement of $1$ (``failure") the state of the target qubit is $Z \, |\psi \rangle$ and thus the effect of the circuit can be reversed at zero cost in terms of $T$-count.

Since the minimum $T$-count of the Clifford+$T$ representation of $V^{\dagger}$ is equal to $t$, as per Thm \ref{thm:block:diagonal}, $U$ can be represented exactly and constructively by a two-qubit Clifford+$T$ circuit of $T$-count at most $t+9$.
\end{proof}

In practice we can further reduce the cost by modifying the operator $U$ in the proof of the Corollary \ref{corol:RUS:synthesis} if needed.
We note, for instance, that replacing $U$ with $U=\left[\begin{array}{cc}
              V & 0 \\
              0 & S^{d_1} \, V^{\dagger} S^{d_2}
            \end{array}\right],$
where $d_1,d_2 \in \{-1,0,1\}$ does not change the resulting target state on measurement outcome $0$ (or the probability of that measurement).
However on measurement outcome $1$ we now have $S^{-d_1} \, |\psi \rangle$ on the target qubit.
This effect can be reversed at zero $T$-count.
Replacing $V^{\dagger}$ with $S^{d_1} \, V^{\dagger} S^{d_2}$ often lowers the number of $T$ gates by as much as $5$.
Although in practice performing an $S^{d}$ correction is believed to be slightly more expensive than performing $Z$, this small difference is more than offset by the $T$-count reduction.

\section{Reducing controlled-Pauli Gates in an RUS design}\label{sec:circuit:optimize}

A Repeat-Until-Success circuit, synthesized using Corollary \ref{corol:RUS:synthesis}, may contain a large number of controlled-Pauli gates.
We find that half of the Pauli gates can be eliminated from a decorated $T$ code using a set of signature-preserving rewrite rules.

Using the above algorithms to represent $V^{\dagger}$ as a Pauli decoration of the representation of $V$ results in the same number of controlled-Pauli gates as Pauli gates injected by rewrite rules listed in the proof of Corollary \ref{corol:RUS:synthesis}.
In the worst case, this could lead to an RUS circuit design with $2\,t$ controlled-Pauli gates, where $t$ is the circuit $T$-count.
Although a controlled-Pauli gate has zero $T$-count, a large number of them could be a concern in some hardware architectures.

We find that half of the Pauli gates can be eliminated from a decorated $T$ code using a set of signature-preserving rewrite rules.
\begin{defin}
Consider a single-qubit Clifford+$T$ circuit of the form $g_0 \, T^{k_1} \, g_1 , \ldots , T^{k_r} \, g_r$, where $g_0,g_r$ are Clifford gates and for
$i=1,\ldots,r-1$, $g_i \in \{Id,X,Y,Z,H,X\, H, Y\,H, Z\,H, H\,X, H\,Y, H\,Z\}$.
We call the sequence of integers $\{k_1, \ldots, k_r\}$ the \emph{signature} of such circuit.
\end{defin}

By definition, Pauli decoration of a given circuit is not unique, and functionally any of the equivalent Pauli decorations works equally well in the context of Thm \ref{thm:block:diagonal}.
It is only the signature of the Pauli decoration circuit that matters and we can apply any equivalent transformations to such a decoration as long as those transformations preserve the signature.

\begin{observ}
If the $T$-count of a decorated $T$ code $c$ is $t$, then the circuit can be rewritten in a systematic way to a form $c' \, i^m, m \in \mathbb{Z}$ , where $c'$ is an equivalent decorated $T$ code with no more than $t+1$ Pauli gates and the same signature at $c$.

\end{observ}

Per the following commutation rules $X \, H = H \, Z$; $Y \, H = - Y \, H$; $Z \, H = H \, X$, any term of the form $P \, H \, Q$, where $P,Q$ are Pauli gates, can be reduced to $H \, P' \, i^m$, where $m \in \mathbb{Z}$ and $P'$ is a Pauli gate.
Suppose $c' \, i^m$, $m \in \mathbb{Z}$, is a decorated $T$ code equivalent to circuit $c$ with the same signature as $c$ and it has the minimal number of Pauli gates among those $T$ codes that have the same signature as $c$.
Suppose this minimal number of Pauli gates is greater than $t+1$.
Then by the pigeonhole principle, some occurrence of $H$ in $c'$ has two Pauli gates adjacent to it, at least one of which can be eliminated (without changing the signature of the circuit). But this would contradict the minimality of the Pauli gate count.



\section{An Example} \label{sec:example}

We begin with an example of our algorithm for the target rotation $R_z(\pi/64)$ and precision $\varepsilon = 10^{-11}$.

In Stage 1, the integer relation solver produces a $z^*/z$ approximation of $e^{ i \, \pi/64}$ with $z=1167\, \omega^3 \,-218 \, \omega^2 \, -798 \, \omega \, -359$ and better precision of approximately $3 \times 10^{-12}$.
In Stage 2, the randomized search algorithm inflates $z$ to $-603 \, \omega^3 \, +1694 \, \omega^2 \, -1510 \, \omega \, -7501$ in order to find a solvable norm equation and to achieve a one-round success probability of $0.9885$.
In Stage 3,
the new $z$ is expanded into the single-qubit unitary
$\frac{1}{{\sqrt{2}}^{26}} \, \left[\begin{array}{cc}
              z & y \\
              -y^* & z^*
            \end{array}\right],
$
where $y=1973 \, \omega^3 \, -860 \, \omega^2 \, +358 \, \omega \, +755$.

In Stage 4, the circuit synthesis algorithm generates the following RUS circuit design:

$\Lambda( S^{\dagger} \, H) \,\mathcal{C} \, \Lambda( S^{\dagger} \, H)$,

 where $\mathcal{C}= \omega^3 \, \Lambda(X) \, (T^{\dagger} H \, T^{\dagger} H)_2  \Lambda(X) \, (T^{\dagger} H \, T \, H)_2 $ $ \Lambda(X) \, (T \, H \, T \, H \, T^{\dagger} H)_2
\Lambda(Y) \, (T^{\dagger} H \, T^{\dagger} H)_2  \Lambda(Y) \, (T^{\dagger} H \, T\, H)_2 $ $ \Lambda(Y) \, (T\, H \, T^{\dagger} H \, T\, H \, T^{\dagger} H\, T\, H \, T^{\dagger} H)_2
 \Lambda(X) \, (T^{\dagger} H \, T \, H)_2 $ $ \Lambda(X) \, (T^{\dagger} H \, T \, H)_2 \Lambda(Y) \, (T^{\dagger} H \, T^{\dagger} H)_2
 \Lambda(X) \, (T^{\dagger} H \, T^{\dagger} H)_2 $ $ \Lambda(Y) (T \, H\, T^{\dagger} H \, T \, H)_2 \Lambda(Y) (T \, H\, T \, H)_2
\Lambda(X) \, (T^{\dagger} H \, T \, H)_2 $ $ \Lambda(X) \, ( T \, H \,T^{\dagger} H )_2 \Lambda(X) \, ( T \, H \,T^{\dagger} H \, T \, H )_2
 \Lambda(X) \, (T \, H \, T \, H )_2 $ $ \Lambda(Y) \, (T \, H \, T \, H )_2 \Lambda(Y) \, (T \, H \, T^{\dagger} \, H )_2  \Lambda(Y) \, (T^{\dagger} \, H \, T^{\dagger} \, H $ $\, T \, H \, T^{\dagger} \, H )_2  \Lambda(X) \, (T\, H \, T \, H)_2 $ $ \Lambda(X) \, (T^{\dagger} )_2 \Lambda(X) \, (H \, S)_2$,
and $(O)_2$ denotes $Id \otimes O$ for an operator $O$.

The design cost of this RUS circuit is $58$ $T$ gates.
The expected $T$-count upon success is $58/0.9885 < 58.7$.
We also note that the construction $\Lambda(H) \, \mathcal{C} \, \Lambda(H)$ requires the Clifford $S$ gate for failure correction (as outlined in the end of Appendix \ref{app:synthesis}).
The expected $T$-cost of this construction is $54.6$.
The final trace distance between $R_z(\pi/64)$ and the rotation implemented by the RUS circuit is $1.056 \times 10^{-12}$, which is better than requested.

The expected $T$-count of the RUS circuit in this example is roughly $2.4$ times smaller than the theoretical bound on the $T$-count for the ancilla-free unitary Clifford+$T$ solution, for $\varepsilon = 1.056 \times 10^{-12}$.
This is just shy of our average reduction factor of over $2.5$. which is typical of coarser precisions such as $10^{-12}$.
The reduction factor reliably exceeds $2.5$ for finer precisions.
However, we chose to display a coarser precision example here to make it easier for the reader to parse.

\section{Cost and Performance Evaluation}\label{subsec:runtime}

We have summarized the $T$-count  vs. precision performance of our algorithm on random axial rotations in this work. Here we explore a set of special target rotations known as Fourier rotations.
We also determine the runtime performance of our algorithm.

\subsection{Cost Evaluation on Fourier angles}

The second set consists of $40$ angles of the form $\pi/2^k, k=2,\ldots,41$, as found, for example, in the Quantum Fourier Transform.
We evaluate the second set at $30$ target precisions $\varepsilon \in \{10^{-11}, \ldots, 10^{-40}\}$.
For each precision and target rotation, the algorithm produces the two-qubit RUS circuit that implements the given rotation to precision $\varepsilon$ and we record the circuit's expected $T$-count and the algorithm's runtime.

Figure \ref{fig:fourier:angles} plots the precision $\varepsilon$ versus the mean (and standard deviation) expected $T$-count across the RUS circuits generated for the set of $40$ Fourier angles.
\begin{figure}[t]
\includegraphics[width=3.5in]{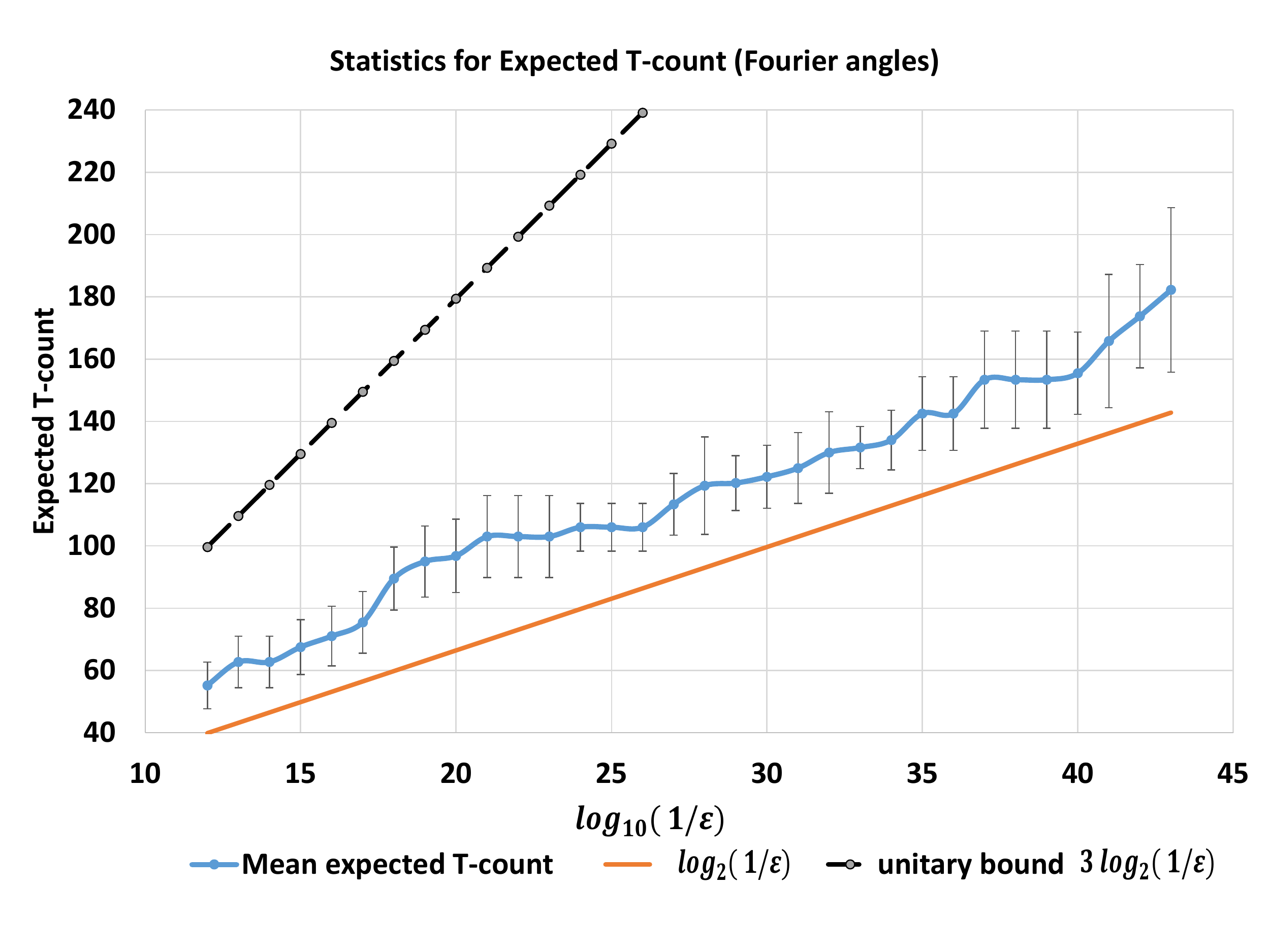}

\caption{\label{fig:fourier:angles} Precision $\varepsilon$ versus the mean expected $T$-count for the set of Fourier angles. } 
\end{figure}
The regression formula for the mean expected $T$-count is $3.59 \, \log_{10}(1/\varepsilon) + 17.5 = 1.08  \, \log_{2}(1/\varepsilon) + 17.5$.
The plateaus in the plot reflect an ``overclocking" phenomenon where the algorithm delivers an RUS circuit with precision better than requested (e.g.,  precision $\sim 10^{-23}$ when $10^{-21}$ is requested).
For reasons that remain to be understood, the overclocking is observed more frequently in the case of angles of the form $\pi/2^k$ than for completely random angles.

\subsection{Runtime Performance Evaluation}

Our algorithm is implemented in \emph{Mathematica} \footnote{\emph{Mathematica} is a registered trademark of Wolfram Research} and uses a customization of PSLQ code downloaded from \cite{PSLQBertok}.
We also use \emph{Mathematica}'s TimeConstrained integer factorization for deciding if a norm equation is easily solvable.
As a result, we have not optimized our code for clock speed and thus assess, instead of wall clock speed, the number of iterations of the PSLQ integer relation algorithm required on average for various precisions.
We also asses the number of instances of the norm equation our synthesis algorithm solves on average per RUS circuit synthesized.
The PSLQ-based cyclotomic approximation and solving of the norm equation are the only intense computational blocks in our algorithm; the combined time of all the other steps in the algorithm is trivial.

The main Theorem of \cite{FergBail} states that if exact integer relations between the subject real values exist  and $M$ is the minimum norm of such an integer relation then the PSLQ algorithm terminates after a number of integration bounded by $O(\log(M))$.
Both the Theorem and the proof can be modified to apply to our customization of the algorithm that looks for approximate integer relations, to state that if $M_{\varepsilon}$ is the minimum size of an integer vector $a$ such that $|a \, x| < \varepsilon$ then the modified algorithm terminates after a number of iterations bounded by $O(\log(M_{\varepsilon}))$.
As we have noted in Observation \ref{observe:size}, $M_{\varepsilon} = O(\varepsilon^{-1/4})$ and thus the bound on the number of iterations to termination is linear in $\log(\varepsilon)$.

The empirical PSLQ iteration statistics obtained by simulation are presented in Figure \ref{fig:PSLQ:iterations}.
The regression formula for the precision $\varepsilon$ versus number of iterations is estimated as $iterations = 3.86 \, \log_{10}(1/\varepsilon) - 6.77$.
Although the expected number of iterations is confirmed to be linear in $\log(1/\varepsilon)$, we need to keep in mind that the PSLQ algorithm requires variable precision arithmetic where the size of the mantissa also grows linearly with $\log(1/\varepsilon)$, so the practical cost in terms of elementary big-digit operation turns out to be quadratic in $\log(1/\varepsilon)$.

\begin{figure}[t]
\includegraphics[width=3.5in]{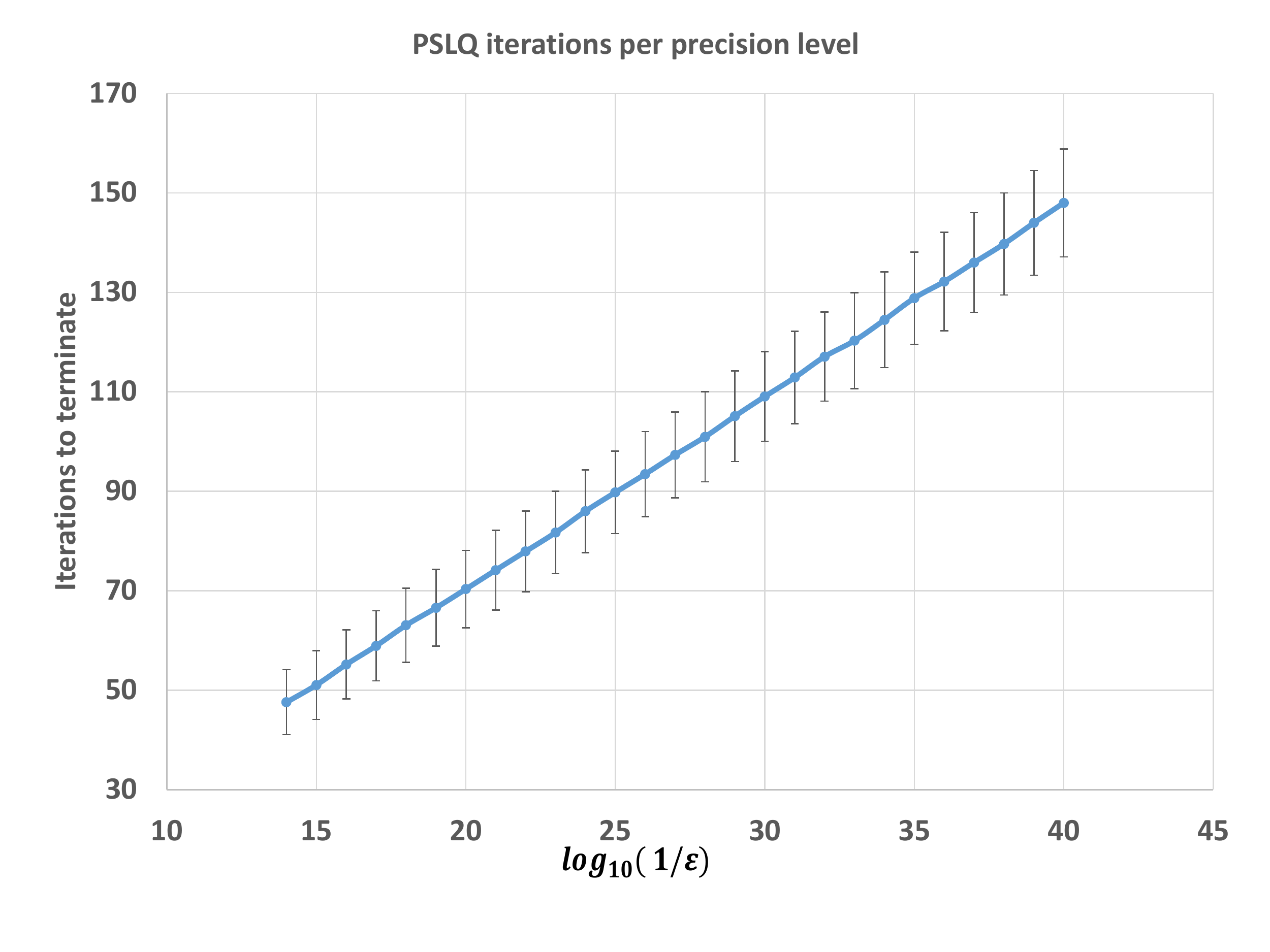}

\caption{\label{fig:PSLQ:iterations} Mean expected iterations to termination of the PSLQ algorithm. } 
\end{figure}

The cost of solving norm equations in our context is well researched and documented (c.f.,~\cite{DSimon,LWashington,Selinger,UGPARI}).
Apart from expensive general factorization, it has been proven to be probabilistically polynomial in the bit size of the right-hand side (RHS) of the norm equation.
Since we pick only instances of the RHS that are easy to factor and since the magnitudes of the RHS are sublinear in $1/\varepsilon$, our norm equation cost remains probabilistically polynomial in $\log(1/\varepsilon)$.
It is easy to see that our randomized search for modifiers is structured such that the number of norm equations we are solving per target per precision level is polynomial in $\log(1/\varepsilon)$.
Thus the overall total runtime cost of our algorithm is probabilistically polynomial in $\log(1/\varepsilon)$.

\section{Generalizations}\label{sec:general}

Whereas so far we studied the case of RUS designs for single-qubit unitary operations---and notably the special case of axial rotations---using a single ancilla qubit, in this section we present several results about RUS designs in case several ancilla qubits are available. Furthermore, some of the results presented in this section will hold for RUS designs of multi-qubit unitary operations also.

Our first result, shown in Appendix \ref{sec:2anc}, asserts that an arbitrary unitary $V$ can be implemented {\em exactly} by a RUS circuit using at most $2$ additional qubits, provided that there exists a pre-multiplier that turns all matrix elements of $V$ into elements of the ring $\Z[i,\frac{1}{\sqrt{2}}]$.
This condition is always met in case the unitary is defined over the field $\Q(\omega)$, allowing us to implement any such unitary using an RUS design with at most $2$ ancilla qubits.

In Appendix \ref{sec:1anc} we consider a possible strengthening of this result by considering RUS designs that only allow at most $1$ additional qubits.
In this case it turns out that, assuming a conjecture from number theory, one ancilla qubit can be saved, i.e., the same class of unitaries can be implemented as in Appendix \ref{sec:2anc}.

Finally, the third result established in Appendix \ref{sec:3anc} is an alternative to the implementation of the ``Jack-of-Daggers'' matrices introduced in Lemma \ref{lem:JOD} which has the advantage of avoiding additional entangling gates and tightly controlling the $T$-depth of resulting unitary transformation in terms of the $T$-depth of the unitary transformation $V$.
The design uses two additional ancillas, increases the $T$-count by a factor of $3$, and is applicable in case $V$ is a rotation around an axis which lies in the $xy$ plane.

\subsection{Exact implementations with 2 ancillas}\label{sec:2anc}

\begin{thm}\label{th:RUS}
Let $n\geq 1$ and let $V\in \C^{2^n\times 2^n}$ be unitary. If there exists $\alpha\in \C$ such that $\alpha V \in \Z[i,\frac{1}{\sqrt{2}}]^{2^n\times 2^n}$ and $|\alpha|^2 \in \Z$, then there exists a RUS circuit to implement $V$ with at most $n+2$ qubits. The two ancilla qubits can be assumed to be initialized in state $\ket{0}$. Moreover, if the RUS circuit fails, it implements the identity $\onemat_{2^n}$.
\end{thm}

To establish the theorem we will need a simple lemma about unitary embeddings (see, e.\,g., \cite[Theorem 6.9]{Zhang:2011}).

\begin{lem}\label{lem:embed}
Let $N \geq 1$ and let $A\in \C^{N \times N}$ be a contraction, i.e., $\|A\|\leq 1$ holds with respect to the spectral norm $\|\cdot\|$. Then there exists a unitary embedding $U\in \C^{2N \times 2N}$ of $A$ of the form
\[
U  = \left[\begin{array}{cc}
A & (\onemat-A A^\dagger)^{1/2}\\
(\onemat-A^\dagger A)^{1/2} & -A^\dagger
\end{array}
\right].
\]
\end{lem}
We will apply Lemma \ref{lem:embed} only for the special case where $A=\alpha V$ is a scaled unitary, where $\alpha\in \C$. Clearly, if $U$ is a unitary embedding of $A$ as above, then also $(\onemat_N \oplus S) U (\onemat_N \oplus T)$ is a unitary embedding for $A$ whenever $S$ and $T$ are unitary. We will use this additional degree of freedom for $S$ and $T$ being scalar multiples of the identity that are conjugates to each other, i.e., $S=\lambda \onemat_N$, $T=\overline{\lambda} \onemat_N$, where $|\lambda|^2\leq 1$. Putting these two considerations and the lemma together, this means we consider unitary embeddings of the form
\[
U' =  \left[\begin{array}{cc}
\alpha V & \overline{\beta} \onemat \\
\beta \onemat & -\overline{\alpha} V^\dagger
\end{array}
\right],
\]
where $|\alpha|^2+|\beta|^2=1$. Here and in the following $\overline{x}$ denotes the complex conjugate of a complex number $x$, and $\onemat$ and $\zeromat$ denote the identity and the all-zero matrix of appropriate size.

The main idea of the proof is to a) reduce the problem of finding suitable scaling factors $\alpha, \beta$ to quadratic Diophantine equations and then to b) apply Lemma \ref{lem:embed} twice in order to have enough degrees of freedom in the Diophantine equations so that its solution can be obtained from the four squares theorem.

\begin{proof} Of Theorem \ref{th:RUS}.

Choose an integer $\ell\geq 0$ such that $\alpha 2^{-\ell/2} \leq 1$. Clearly such an integer exists, we can, e.g., set $\ell = \lceil 2 \log_2(\alpha) \rceil$. Now the matrix $A := \alpha 2^{-\ell/2}V$ is a contraction as $A A^\dagger = |\alpha|^2 2^{-\ell} V V^\dagger = t \onemat_{2^n}$, where $t \leq 1$ by construction.
We will now apply Lemma \ref{lem:embed} twice: once for $A$ to get a unitary embedding into a unitary of twice the size as $A$ and then again to a scaled version of this unitary to get a matrix of four times the size of $A$. Overall, we obtain that $W$ as following is a unitary embedding of $V$:
\begin{equation}\label{eq:four}
W = \left[\begin{array}{cccc}
\frac{\alpha}{2^{\ell/2}} V & \overline{\beta} \onemat_{2^n} & \overline{\gamma} \onemat_{2^n} & \zeromat_{2^n} \\
\beta \onemat_{2^n} & -\frac{\overline{\alpha}}{2^{\ell/2}} V^\dagger & \zeromat_{2^n} & \overline{\gamma} \onemat_{2^n} \\
\gamma \onemat_{2^n} & \zeromat_{2^n} & -\frac{\overline{\alpha}}{2^{\ell/2}} V^\dagger & -\overline{\beta} \onemat_{2^n} \\
\zeromat_{2^n} & \gamma \onemat_{2^n} & -\beta \onemat_{2^n} & \frac{\alpha}{2^{\ell/2}} V
\end{array}\right],
\end{equation}
provided that $\frac{|\alpha|^2}{2^\ell} + |\beta|^2 + |\gamma|^2 = 1$. In order to find suitable $\beta$, $\gamma$, we choose the following ansatz: let
$\beta=\beta_0/2^{\ell/2}$ and $\gamma=\gamma_0/2^{\ell/2}$ where $\ell$ is as above and $\beta_0,\gamma_0 \in \Z[i]$, i.e., say, $\beta_0 = a+bi$ and $\gamma_0 = c+di$ where $a,b,c,d \in \Z$.

Then by clearing the denominator we find that solving the equation $|\alpha|^2 + |\beta_0|^2 + |\gamma_0|^2 = 2^{\ell}$ is sufficient for the existence of an embedding as in Eq~(\ref{eq:four}). We now apply the four squares theorem to determine integers $a,b,c,d$ such that
\[
a^2 + b^2 + c^2 + d^2 = 2^{\ell}-|\alpha|^2
\]
holds which establishes the unitary embedding. Notice furthermore that due to the assumption on $\alpha$ and $V$ and due to the choice of $\beta$, $\gamma$, we find that the coefficients of matrix $W$ in Eq~(\ref{eq:four}) are elements of $\Z[i,\frac{1}{\sqrt{2}}]$ and that $W$ is unitary.
Using \cite[Theorem 1]{GilSel} we can implement $W$ exactly over Clifford$+T$ using at most one additional ancilla qubit. What is more, as the determinant $\det(M)$ of a block matrix $M = \left[\begin{array}{rr} A & B \\ C & D \end{array}\right]$ can be computed as $\det(M) = \det(AD-CD)$, we compute that $\det(W)=1$ for the matrix $W$ as in Eq~(\ref{eq:four}) holds. Hence \cite[Lemma 7]{GilSel} can be used that establishes that we can avoid this additional ancilla qubit.

Finally, we note that all ancillas in the above construction were chosen to be initialized to $\ket{0}$ and the correction operation in case the circuit fails is always just a global phase times the identity, i.e., the circuit is an RUS circuit that implements $V$ using at most $2$ ancillas.
\end{proof}

\begin{corol}
Let $V \in \C^{2^n \times 2^n}$ be a unitary that can be implemented exactly over Clifford$+T$ using a RUS design with $k\geq 2$ ancilla qubits initialized to $\ket{0}$ and Clifford corrections, then there exists a RUS design with at most $2$ ancilla qubits initialized to $\ket{0}$ where in the ``failure'' case no correction is needed.
\end{corol}
\begin{proof}
This is immediate as the existence of a RUS circuit implies the existence of $\alpha\in \C$ such that $\alpha V$ satisfies the assumption of Theorem \ref{th:RUS}.
\end{proof}

\begin{corol}
Let $n \geq 1$ and let $V\in \Q(\omega)^{2^n \times 2^n}$ be a unitary matrix on $n$ qubits, where as before $\omega = \exp(\pi i/4)$ is a primitive $8$th root of unity. Then $V$ can be implemented exactly by a RUS circuit using at most $2$ ancilla qubits.
\end{corol}
\begin{proof}
Denoting by $v_{i,j}=\frac{a_{i,j}}{b_{i,j}}$ the entries of $V$ where $a_{i,j}$, $b_{i,j} \in \Q(\omega)$, we observe that $v_{i,j}$ can always be written in the form $v_{i,j} = a^\prime_{i,j} b^\prime_{i,j}$, where $a^\prime_{i,j} \in \Z[\omega]$ and $b^\prime_{i,j} \in \Z$. Indeed, $v_{i,j} = a_{i,j} \prod_{\sigma \not= 1: \in {\rm Gal}(\Q(\omega)/\Q)} b_{i,j}^\sigma/ N_{\Q(\omega)/\Q)}(b_{i,j})$ is such a representation. Letting $\alpha :={\rm lcm}(b^\prime_{i,j} : i,j=0\ldots 2^n-1)$ we can apply Theorem \ref{th:RUS} to obtain the stated result.
\end{proof}

\subsection{Exact implementations with $1$ ancilla}\label{sec:1anc}

An immediate question is whether the number of ancillas can be reduced, i.e., if indeed $2$ ancillas are required. We were not able to establish this, but we sketch in the following a heuristic argument that in most cases $1$ ancilla will suffice.

The idea is to consider a unitary embedding of the form
\[
U = \left[\begin{array}{cc}
\theta \alpha V & \overline{\lambda} \onemat \\
\lambda \onemat & -\overline{\theta} \overline{\alpha} V^\dagger
\end{array}
\right]
\]
where $\theta$, $\lambda \in \Q[\omega]$ such that the condition
$|\theta|^2 |\alpha|^2 + |\lambda|^2 = 1$ holds.
We make some additional assumptions, focusing on the case where $V$ is a unitary over $\Q(\omega)$. Then $\alpha$ can be chosen to be an integer and it can furthermore be shown that without loss of generality, we can choose $\alpha$ to be odd. We fix this choice in the following. The above condition can then be restated as a norm form equation involving $N = N_{\Q(\omega)/\Q_+}$ which denotes the relative norm from the field $\Q(\omega)$ of cyclotomic $8$th roots to its maximal real subfield $\Q_+ = \Q(\sqrt{2})$. It is known that the norm equation $N(\theta)=p$ is solvable for $p$ prime if and only if $p \not\equiv 7 \mod  8$. Our ansatz is to solve the norm form equation
\begin{equation}\label{eq:dio}
\alpha^2 N(\theta) + N(\lambda) = 2^\ell
\end{equation}
which can be restated as a system of (exponential) Diophantine equations for $\ell$ and $a,b,c,d,a^\prime,b^\prime,c^\prime,d^\prime$ where $\theta=a+b\omega+c\omega^2+d\omega^3$, $\lambda=a^\prime+b^\prime \omega +  c^\prime \omega^2 + d^\prime \omega^3$. Again, we make an ansatz, namely we assume that $N(\theta)=(8k+3)$, where $k \in \Z$. By considering
the set $S_\alpha = \{ 2^\ell - \alpha^2 (8k+3) : \ell, k \geq 0\} \cap \N$, we see from a pigeonhole argument that $S_\alpha$ must contain an infinitely long arithmetic progression of the form $n_0 + 8 \alpha^2 k$, $k \geq 0$ with some element $0\leq n_0 < \alpha^2$. Now $(n_0,8 \alpha^2)=1$, as otherwise there would be a common prime divisor $p>2$ which cannot exist as $n_0 +  8\alpha^2 k_0 = 2^\ell$ for some $k_0 \geq 0$. Hence, we can apply Dirichlet's theorem which establishes the existence of infinitely many primes $p$ in the set $S_\alpha$. By construction, all of these primes satisfy $p \equiv 5  \mod 8$, i.e., we can solve the form equation $N(\lambda)=p$ for all these primes.

Using our ansatz, we have therefore reduced the solvability of Eq~(\ref{eq:dio}) to a conjecture due to Dickson who conjectured that for $n$ arithmetic progressions of the form $a_i k + b_i$, which all have to satisfy the Dirichlet condition, there exist infinitely many {\em simultaneous} primes, i.e., values $k$ such that all progressions produce a prime number. Unfortunately, Dickson's conjecture is unproven and even establishing the existence of a single simultaneous prime is a long-standing open problem in number theory.

A heuristic argument shows however, that such primes are likely to exist: using Chebotar$\ddot{{\rm e}}$v's density theorem \cite{Chebotarev:1926}, we obtain that there is a constant probability that among the numbers $(8k+3)$, $k \leq M$---for $M$ sufficiently large---we find one that is prime. Of course, this does not guarantee the existence of such a prime as the events for $2^\ell - \alpha^2 (8k+3)$ and $(8k+3)$ to be simultaneous prime are not independent. However, in practice we find that for a given $\alpha$ one quickly finds suitable pairs $(\ell, k)$ corresponding to simultaneous primes and thereby to solutions of Eq~(\ref{eq:dio}).

Finally, we remark that another approach is to study the quadratic form that can be obtained from Eq~(\ref{eq:dio}) by restricting $\theta$, $\lambda$ to be Gaussian integers. Then we arrive at a classic problem of representing integers by quaternary quadratic forms. It would be sufficient to show that the particular form that arises from the above Diophantine equation has a finite exception set (i.e., forms which represent all integers except for a possibly finite set of numbers). This problem has been completely solved by Kloosterman who characterized all quaternary quadratic forms with a finite exception set. Unfortunately, it turns out that the form in Eq~(\ref{eq:dio}) leads to cases where there is an infinite set of exceptions, more precisely, they fall into case 3 in \cite[Section 5.7]{Kloosterman:1926}.

\subsection{Low-depth implementations with $3$ ancillas}\label{sec:3anc}

Lemma \ref{lem:JOD} allows to relate the $T$-count of a single-qubit unitary $V$ to the $T$-count of a $2$-qubit unitary of the specific ``Jack of Daggers'' form
\[
J(V) = \left[\begin{array}{rr}
V & 0 \\
0 & V^\dagger
\end{array}\right]
\]
by exhibiting a circuit that has a moderate increase in $T$-count compared to $V$ and does not require any ancillary qubits for its implementation. However, the number of entangling operations that might be needed to realize $J(V)$ can grow linearly with the $T$-count of $V$.

In this section we present an alternative realization for $J(V)$ which, besides a constant additive increase has the same $T$-depth as $V$. The unitaries $V$ for which this implementation is applicable is given by the set
\[
{\cal S}_{xy} = \{V \in SU(2): V = R_z(\varphi) R_y(\theta) R_z(-\varphi), \theta, \varphi \in \R\},
\]
where $R_y(\theta) = e^{-i\theta Y}=\left[\begin{array}{rr} \cos{\theta} &  -\sin{\theta} \\ \sin{\theta} & \cos{\theta}\end{array}\right]$ and $R_z(\varphi) = e^{-i\varphi Z}=\left[\begin{array}{rr} e^{-i\varphi} & 0 \\ 0 & e^{i \varphi}\end{array}\right]$. Alternatively, we can characterize the unitaries in ${\cal S}_{xy}$ as those operations that have a matrix element $V_{0,0}$ contained in the real numbers. Geometrically, the set of rotations in ${\cal S}_{xy}$ corresponds to the rotations of the Bloch sphere around axes that lie in the $xy$-plane.

We use a circuit that realizes an idea similar to \cite{MS:2001} (see also \cite[Section S8.4]{KitaevEtAl2002} \cite{AMMR:2013}, \cite{SM:2013}) to implement controlled unitaries in cases where an eigenstate of the unitary can be efficiently prepared. We are not in a position to directly implement an eigenstate of $V$, however, an eigenstate (with corresponding eigenvalue of $1$) of $V\otimes V$ can be efficiently prepared and does an eigenstate (again, with corresponding eigenvalue of $1$) of $V^\dagger \otimes V$. Specifically, these eigenstates are members of the familiar set of Bell states.

\begin{lem}
Consider the Bell states $\Psi^{\pm} = \frac{1}{\sqrt{2}}(\ket{01} \pm \ket{10})$. Let $V \in SU(2)$  be an arbitrary single-qubit unitary of determinant one. Then $V \otimes V \ket{\Psi^-} = \ket{\Psi^-}$. Furthermore, let $V\in {\cal S}_{xy}$. Then $V^\dagger \otimes V \ket{\Psi^+} = \ket{\Psi^+}$.
\end{lem}
\begin{proof} The statement $V\otimes V \ket{\Psi^-} = \ket{\Psi^-}$ is the well-known fact that the singlet state has the same representation with respect to any local unitary basis \cite[Section 2.6]{IkeAndMike2000}. For the statement about the unitaries in ${\cal S}_{xy}$, we compute that $V=R_z(\varphi)R_y(\theta)R_z(-\varphi)=\left[\begin{array}{rr}
\cos{\theta} & - e^{-2i\varphi}\sin{\theta} \\
e^{2i\varphi}\sin{\theta} & \cos{\theta}
\end{array}
\right]$ which implies that $V^\dagger \otimes V$ is given by
\begin{widetext}\begin{equation*}
V^\dagger \otimes V = \left[\begin{array}{cccc}
\cos^2{\theta} & -e^{-2i\varphi}\cos{\theta}\sin{\theta} &
e^{-2i\varphi}\cos{\theta}\sin{\theta} & -e^{4i\varphi}\sin^2{\theta}\\
e^{2i\varphi}\cos{\theta}\sin{\theta} & \cos^2{\theta}&
\sin^2{\theta} & e^{-2i\varphi}\cos{\theta}\sin{\theta}\\
-e^{2i\varphi}\cos{\theta}\sin{\theta} & \sin^2{\theta} &
\cos^2{\theta} & -e^{-2i\varphi}\cos{\theta}\sin{\theta} \\
-e^{-4i\varphi} \sin^2{\theta} & -e^{2i\varphi}\cos{\theta}\sin{\theta}&
e^{2i\varphi}\cos{\theta}\sin{\theta} & \cos^2{\theta}
\end{array}
\right]
.
\end{equation*}
\end{widetext}

Hence, applying $V^\dagger \otimes V$ to $\ket{\Psi^+}=\frac{1}{\sqrt{2}}(\ket{01}+\ket{10})$ we see that $V^\dagger \otimes V \ket{\Psi^+} = \ket{\Psi^+}$.
\end{proof}

\begin{thm}
Let $V\in {\cal S}_{xy}$ be a unitary and let $t$ be the $T$-count of $V$. Then $J(V)$ can be implemented over the Clifford$+T$ gate set with a $T$-depth of at most $t+8$.
\end{thm}
\begin{proof}
The proof of the theorem follows from an inspection of the circuit shown in Fig.~\ref{fig:RUSanc}. We first show the correctness of the circuit, i.e., establish that it transforms the input state $\ket{b}\ket{\psi}\ket{0}^{\otimes 2}$ to the output state $\ket{b}A \ket{\psi}\ket{0}^{\otimes 2},$ where $A = V^{(-1)^b}$, which is equal to $(J(V)\ket{b}\ket{\psi})\otimes\ket{0}^{\otimes 2}$. We then show the claimed statement about the $T$-count of the circuit.
The circuit uses $2$ ancillas, in addition to the $1$ ancilla needed for the RUS design. The ancillas are initialized in the $\ket{0}$ state.

\begin{figure}[tb]
\bigskip
\centerline{
\includegraphics[width=\columnwidth]{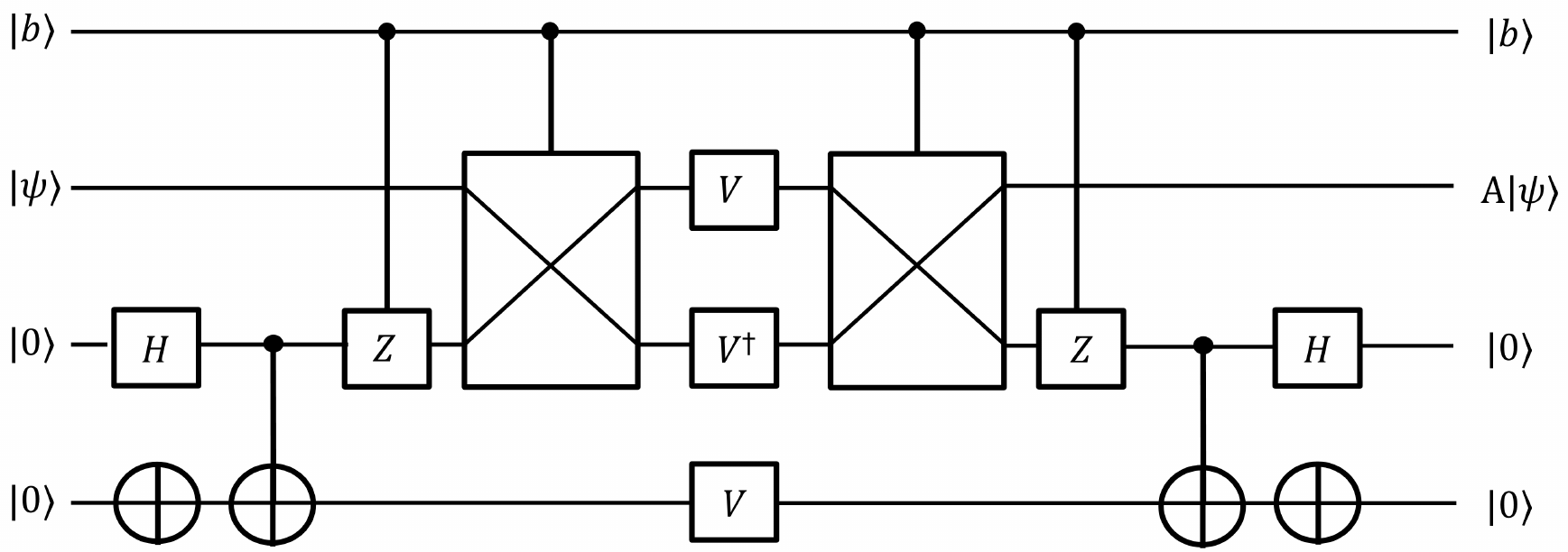}
}
\caption{\label{fig:RUSanc} Implementation of a matrix $J(V)$ using $2$ additional ancillas. The input is a qubit $\ket{b}$ that indicates whether the overall transformation $A$ applied to the data qubit $\ket{\psi}$ at the output is either $V$ or $V^\dagger$. The circuit uses $2$ ancillas, in addition to the $1$ ancilla needed for the RUS design. The gates in the center of the circuit $\onemat_2 \otimes V \otimes V^\dagger \otimes V$ can be executed in parallel, i.e., up to a constant increase due to the controlled-SWAP gates, the overall $T$-depth of the circuit is the same as the $T$-depth of $V$.
}
\end{figure}

It is easy to see that the first four gates prepare the state $\ket{b}\ket{\psi} \ket{\Psi^{(-1)^b}}$. We now consider two cases: if $b=0$, then the application of the following three gates, i.e., a controlled-SWAP, followed by $\onemat_2 \otimes V \otimes V^\dagger \otimes V$, followed by a controlled-SWAP  maps $\ket{0}\ket{\psi}\ket{\Psi^+}$ to $\ket{0}V \ket{\psi} \ket{\Psi^+}$ as $\ket{\Psi^+}$ is an eigenstate of $V^\dagger \otimes V$ of eigenvalue $1$. If on the other hand $b=1$, then the application of the same three gates maps $\ket{0}\ket{\psi}\ket{\Psi^-}$ to $\ket{0}V^\dagger \ket{\psi} \ket{\Psi^-}$ as $\ket{\Psi^-}$ is an eigenstate of $V \otimes V$ of eigenvalue $1$.
Hence, depending on the value of $b$, the resulting Bell state is then ``routed'' to either $V^\dagger\otimes V$ (in case $b=0$) or $V \otimes V$ (in case $b=1$) and correspondingly, the data qubit is routed to either $V$ (in case $b=0$) or $V^\dagger$ (in case $b=1$). Finally, the remainder of the circuit cleans up the ancillas and brings them back into the $\ket{0}$ states.

Regarding the complexity of the circuit in terms of $T$-gates, we use
the implementation of the controlled-SWAP gate given in \cite{AMMR:2013} which uses a total of $7$ $T$ gates and has a $T$-depth of $4$. The application of $\onemat_2\otimes V \otimes V^\dagger \otimes V$ has the same $T$-depth $t$ as the application of the single-qubit gate $V$. As there are two controlled-SWAP gates needed and the remaining gates in the circuit are Clifford gates, i.e., they do not introduce additional $T$-gates, the claimed bound of $t+8$ for the total $T$-depth follows.
\end{proof}

\end{document}